\documentclass[a4paper,11pt]{article}
\usepackage[margin=1in]{geometry}
\usepackage[utf8]{inputenc}
\usepackage[numbers,square,comma,sort]{natbib}

\usepackage[T1]{fontenc}
\usepackage{graphicx}
\usepackage{amssymb}
\usepackage{amsmath}
\usepackage{amsthm}
\usepackage{hyperref}
\hypersetup{
    colorlinks=true,
    linkcolor=blue,
    filecolor=magenta,
    urlcolor=cyan!60!black,
    citecolor=green!60!black,
}
\usepackage{tikz}
\usetikzlibrary {arrows.meta}
\usetikzlibrary{calc}
\usetikzlibrary{backgrounds}
\usetikzlibrary{positioning}

\usepackage{bm}
\usepackage{subcaption}
\usepackage[linesnumbered,algoruled,boxed,lined]{algorithm2e}
\usepackage{complexity} 
\usepackage{thmtools}
\usepackage[shortlabels]{enumitem}
\usepackage{thm-restate}
\usepackage[capitalize,noabbrev]{cleveref}

\newcommand{\eMC}{\textsc{Edge Multicut}\xspace} 
\newcommand{\wMCshort}{\textsc{wMC}\xspace}

\newcommand{\wmctreeslong}{\textsc{Unconstrained Weighted Multicut on Tree}\xspace}
\newcommand{\wmctrees}{\textsc{uwMC-Tree}\xspace}

\newcommand{\wmc}{\textsc{uwMC}\xspace}

\newcommand{\bimctreeslong}{\textsc{Weighted Multicut on Tree}\xspace}
\newcommand{\bimctrees}{\textsc{wMC-Tree}\xspace}

\newcommand{\wdpc}{\textsc{wDPC}\xspace}
\newcommand{\wdpclong}{\textsc{Weighted Digraph Pair Cut}\xspace}

\newtheorem{branching rule}{Branching Rule}[section]

\newcommand{\nph}{\textsf{\textup{NP}}-hard\xspace}

\newcommand{\timewmctrees}{$2^{\OO(\ell^2 \log \ell)} \cdot n^{\OO(1)}$}

\newcommand{\timebimctrees}{$2^{\OO(\kk^4)} \cdot n^{\OO(1)}$}

\newcommand{\timewmctreesdl}{$3^d \cdot 2^{d q} \cdot 2^{\OO(q^2 \log q)} \cdot n^{\OO(1)}$}

\newcommand{\timespecialdpc}{$2^{\OO(\ell^2 \log \ell)} \cdot n^{\OO(1)}$}

\newcommand{\timedpc}{$2^{\OO(\kk^4)} \cdot n^{\OO(1)}$}

\newcommand{\timedpcarcless}{$2^{\OO(\ell^2 \log \ell)} \cdot n^{\OO(1)}$}

\newcommand{\Adpc}{\ensuremath{\mathcal{A}_{\textsf{dpc}}}\xspace}
\newcommand{\Apath}{\ensuremath{\mathcal{A}_{\textsf{path}}}\xspace}
\newcommand{\Astar}{\ensuremath{\mathcal{A}_{\textsf{star}}}\xspace}
\newcommand{\Aunmc}{\ensuremath{\mathcal{A}_{\textsf{un-mc}}}\xspace}

\newcommand{\Tab}{\ensuremath{\textsf{Tab}}\xspace}
\newcommand{\Aheavyv}{\ensuremath{\mathcal{A}^{\textsf{heavy}}_v}\xspace}

\newcommand{\Pairs}{\mathcal P}
\DeclareMathOperator{\light}{\mathsf{light}} 

\newcommand{\wt}{\texttt{wt}}

\newcommand{\kk}{k}
\newcommand{\rr}{\mathbf{r}}
\newcommand{\ww}{\mathbf{w}}

\usepackage{color}

\newcommand{\blue}[1]{{\color{blue}{#1}}}

\newcommand{\pt}[1]{\textcolor{magenta}{#1 - PT}}

\newcommand{\calI}{{\mathcal I}}

\newcommand{\OO}{{\mathcal O}}
\newcommand{\calP}{\mathcal{P}}

\newcommand{\yes}{\textsc{Yes}\xspace}
\newcommand{\no}{\textsc{No}\xspace}

\newtheorem{theorem}{Theorem}[section]
\newtheorem{lemma}[theorem]{Lemma}
\newtheorem{corollary}[theorem]{Corollary}

\newtheorem{proposition}[theorem]{Proposition}
\newtheorem{reduction rule}[theorem]{Reduction Rule}
\newtheorem{marking-scheme}[theorem]{Marking Scheme}
\newtheorem{definition}[theorem]{Definition}
\newtheorem{remark}[theorem]{Remark}

\newcommand{\defproblem}[3]{
  \vspace{1mm}
\noindent\fbox{
  \begin{minipage}{0.96\textwidth}
  \begin{tabular*}{\textwidth}{@{\extracolsep{\fill}}lr} #1 \\ \end{tabular*}
  {\bf{Input:}} #2  \\
  {\bf{Question:}} #3
  \end{minipage}
  }
  \vspace{1mm}
}

\tikzset{%
sibling distance=3em,
level distance=2em,
node distance=1.1em and 0.4em,
edge from parent/.style = {
  draw=black,
  thick,
  solid,
},
every edge/.style = {
  thick,
},
stan/.style = {circle,
  draw=black, align=center, right,
  fill=white,
  align=left,
  thick,
  solid,
  edge from parent/.style={draw=green},
  scale=0.6,
  },
changed/.style = {
  fill=blue!50!white,
},
stay/.style = {
stays
},
del/.style = {
  red,
  fill=white,
  dotted,
},
contract/.style = {
  draw=orange,
  thick,
  arrows = {Stealth[reversed,scale=0.7]-Stealth[reversed,scale=0.7]},
},
dangling/.style={
},
undel/.style = {
  fill=green!70!black,
}
}%

\begin{document}

\title{Parameterized Complexity of Weighted Multicut in Trees}
%
%
\author{
  Esther Galby\textsuperscript{1}\\
  \texttt{esther.galby@cispa.de}
\and
  D\'aniel Marx\textsuperscript{1}\\
  \texttt{marx@cispa.de}
\and
  Philipp Schepper\textsuperscript{1}\\
  \texttt{philipp.schepper@cispa.de}
\and
  Roohani Sharma\textsuperscript{2}\\
  \texttt{rsharma@mpi-inf.mpg.de}
\and
  Prafullkumar Tale\textsuperscript{1}\\
  \texttt{prafullkumar.tale@cispa.de}
}
\date{
  \normalsize
  \textsuperscript{1}
  CISPA Helmholtz Center for Information Security, Germany
  \\
  \textsuperscript{2}
  Max Planck Institute for Informatics, SIC, Germany
}
\maketitle              
\begin{abstract}
  The \eMC problem is a classical cut problem
  where given an undirected graph $G$, a set of pairs of vertices $\mathcal{P}$,
  and a budget $\kk$,
  the goal is to determine if there is a set $S$ of at most $\kk$ edges
  such that for each $(s,t) \in \mathcal{P}$, $G-S$ has no path from $s$ to $t$.
  \eMC has been relatively recently shown to be fixed-parameter tractable (\FPT),
  parameterized by $\kk$, by Marx and Razgon\ [SICOMP 2014],
  and independently by Bousquet et al.\ [SICOMP 2018].
  In the weighted version of the problem, called \textsc{Weighted Edge Multicut}
  one is additionally given a weight function $\wt : E(G) \to \mathbb{N}$ and a weight bound $\ww$,
  and the goal is to determine if there is a solution of size at most $\kk$ and weight at most $\ww$.
  Both the \FPT\ algorithms for \eMC by Marx et al.\ and Bousquet et al.\ fail to generalize to the weighted setting.
  In fact, the weighted problem is non-trivial even on trees
  and determining whether {\sc Weighted Edge Multicut} on trees is \FPT\ %
  was explicitly posed as an open problem by Bousquet et al.\ [STACS 2009].
  In this article, we answer this question positively
  by designing an algorithm
  which uses a very recent result by Kim et al.\ [STOC 2022] about directed flow augmentation as subroutine.

  We also study a variant of this problem
  where there is no bound on the size of the solution,
  but the parameter is a structural property of the input, for example,
  the number of leaves of the tree.
  We strengthen our results by stating them for the more general vertex deletion version.


\end{abstract}

\section{Introduction}
\eMC is a generalization of the classical {\sc $(s,t)$-Cut} problem
where given a graph $G$,
a set of terminal pairs $\mathcal{P}=\{(s_1,t_1), \ldots, (s_p,t_p)\}$,
and an integer $\kk$,
the goal is to determine if there exists a set of at most $\kk$ edges
whose deletion disconnects $s_i$ from $t_i$, for each $i \in [p]$.
Such a set is called a {\em $\mathcal{P}$-multicut} in $G$.
The case $p=1$ corresponds to the classical \textsc{$(s,t)$-cut} problem.
\eMC is polynomial time solvable for $p \leq 2$~\cite{yannakakis1983cutting}
and is \nph\ even for $p = 3$~\cite{dahlhaus1994complexity}.
From the parameterized complexity point of view,
it was a long-standing open question to determine
if the problem is fixed-parameter tractable (\FPT) parameterized by the solution size.
This question was resolved independently by Marx and Razgon~\cite{marx2014fixed} and Bousquet et al.~\cite{bousquet2018multicut},
proving that the problem is \FPT.
Both algorithms extensively use the notion of important separators,
a technique introduced earlier by Marx~\cite{marx2006parameterized}.
Bousquet et al.~\cite{bousquet2018multicut} additionally use several problem-specific observations
and arguments about the structure of multicut instances,
while Marx and Razgon~\cite{marx2014fixed} formulated the technique of random sampling of important separators,
which found further applications for many other problems
\cite{%
chitnis2017list,chitnis2015directed,chitnis2013fixed,%
kratsch2015fixed,%
lokshtanov2013clustering,lokshtanov2012parameterized%
}.


\paragraph*{Weighted Multicut.}
One drawback of the algorithms using important separators is
that they are essentially based on a replacement argument:
if a subset $X$ of the solution satisfies some property,
then this technique allows us to find a set $X'$
such that $X$ can be replaced with $X'$,
thereby making progress towards fully identifying a solution.
This local replacement argument inherently fails
if the overall solution is also required to satisfy additional properties,
such as minimizing the overall weight,
since replacing $X$ with $X'$ may violate these additional constraints.
Thus, the ideas from the algorithms of Marx and Razgon~\cite{marx2014fixed} and Bousquet et al.~\cite{bousquet2018multicut}
fail to generalize to the edge deletion version of {\sc Weighted Multicut} (\wMCshort)
where we are,  additionally, given a weight function $\wt: E(G) \to \mathbb{N}$
and an integer $\ww$,
and the goal is to determine if there exists a $\mathcal{P}$-multicut in $G$
of size at most $\kk$ and weight at most $\ww$.

\paragraph*{(Weighted) Multicut on Trees.}
\eMC
remains \nph\ on trees~\cite{DBLP:journals/algorithmica/GargVY97}
but can be solved in $\OO(2^{\kk} \cdot n)$-time,  where $n$ is the number of vertices in the input tree,
using an easy branching algorithm~\cite{guo2005fixed}:
for the ``deepest'' $(s_i,t_i)$-path
branch on the deletion of the two edges on this path
which are incident to the lowest common ancestor of $s_i$ and $t_i$.
A series of work shows improvement over this simple running time~\cite{chen2012multicut,kanj2014algorithms}, 
and also the problem 
admits a
polynomial kernel~\cite{BousquetDTY09,chen2012multicut}.
Since the algorithmic approaches for \eMC on trees
are based on greedily finding partial solutions,
they do not generalize to the weighted setting.
 In fact, the question whether \wMCshort on trees is \FPT\ (parameterized by the solution size),
 was explicitly posed as an open problem by Bousquet et al.~\cite{BousquetDTY09}.
 In this article, we answer this question in the positive.


\paragraph*{Flow Augmentation.}
As mentioned earlier, most of the available techniques used to design \FPT\ algorithms,
especially for cut problems, do not work in the weighted setting.
Kim et al. \cite{kim2021directed} recently developed the technique of flow augmentation in directed graphs.
This technique offers a new perspective to design \FPT\ algorithms for cut problems
and positively settles the parameterized complexity of some long standing open problems,
such as {\sc Weighted $(s,t)$-Cut}, {\sc Weighted Directed Feedback Vertex Set} and {\sc Weighted Digraph Pair Cut}.

Our main goal is to use
this technique for the underlying core difficulty
in \wMCshort on trees.
More precisely, we do not use the directed flow augmentation technique as such
but we crucially use the \FPT\ algorithm for \wdpclong\ (\wdpc)
which is one important example of the use of this technique.
The \wdpc\ problem is defined as follows \cite{kim2021directed,KratschW20}:
given a \emph{directed} graph $G$,
a source vertex $\rr \in V(G)$,
terminal pairs $\mathcal{P}=\{(s_1,t_1), \ldots, (s_p,t_p)\}$,
a weight function $\wt : E(G) \to \mathbb{N}$,
a positive integer $\kk$,
the goal is to determine if there exists a set $S$ of at most $\kk$ arcs of $G$
such that $\wt(S)$ is minimum%
\footnote{Though the formal description of the problem in~\cite{kim2021directed} asks
for a solution $S$ with $\wt(S) \leq \ww$, the authors remark that the algorithm in fact finds a minimum weight solution.},
and for each $i \in [p]$,
if $G-S$ has a path from $\rr$ to $s_i$,
then $G-S$ has no path from $\rr$ to $t_i$. Such a set is called a {\em $\mathcal{P}$-dpc} with respect to $\rr$ in $G$.
Kim et al.~\cite[Section~6.1ff]{kim2021directed} showed that \wdpc\ %
can be solved in randomized \timedpc-time.
The randomized running time of this algorithm
is an artifact of the use of the directed flow augmentation procedure
which is randomized.
Apart from this step, all the other steps of the algorithm are deterministic.
Our basic observation is that the algorithm for \wdpc\ can be used to solve a
non-trivial base case of \wMCshort\ in trees:
if there is a vertex $\rr \in V(T)$ such that
all the terminal pair paths of $\mathcal{P}$ pass through $\rr$,
then $S$ is a $\mathcal{P}$-multicut of $T$ if and only if
for all $(s,t)\in \Pairs$,
$S$ intersects the $(\rr,s)$-path or the $(\rr,t)$-path.
This is equivalent to saying that
$S$ is a $\mathcal{P}$-dpc for $T$
(in \wdpc we interpret each edge of $T$ to be directed away from~$\rr$).%
\footnote{
When dealing with undirected graphs,
the flow augmentation restricted to undirected graphs
given by Kim et al.~\cite{kim2021solving} may suffice
to solve \wdpc\ on undirected graphs.
As this problem
is not mentioned explicitly in~\cite{kim2021solving},
we stick to the directed setting.
}



\paragraph*{Edge Deletion vs.~Vertex Deletion.}
In the weighted setting,
the edge deletion version of \wMCshort (on trees)
reduces to its vertex deletion version (on trees),
by subdividing each edge and assigning the weight of the original edge to the newly added vertex corresponding to the edge,
and by setting the weights of the original vertices to $\infty$ (or larger than the weight budget parameter).
Note that such a reduction does not work in the unweighted setting
as the vertex deletion version of {\sc Multicut} in trees
is polynomial time solvable~\cite{DBLP:books/sp/CyganFKLMPPS15}.

\paragraph*{Main Result.}
From now on we only study the vertex deletion version of \wMCshort on trees
which, as mentioned above, is more general than the edge deletion version.
It is formally defined below.


\defproblem{\bimctreeslong\ (\bimctrees)}
{A tree $T$, a collection of terminal pairs $\mathcal{P} \subseteq V(T) \times V(T)$,
a vertex weight function $\wt : V(T) \to \mathbb{N}$,
and positive integers $\ww$ and $\kk$.}
{Does there exist $S \subseteq V(T)$
such that $|S| \leq \kk$, $\wt(S) \leq \ww$,
and $S$ intersects the unique $(s,t)$-path in $T$, for each $(s, t) \in \calP$?}

We set $\wt(S) = \sum_{v \in S} \wt(v)$ for the ease of notation.
We use the \FPT\ algorithm for \wdpc\ (restricted to trees)~\cite[Section~6.1ff]{kim2021directed}
as a subroutine to prove our main result,
namely that \bimctrees is \FPT.

\begin{theorem}
  \label{thm:bi-mc-trees-k}
  \bimctrees can be solved in randomized \timebimctrees\ time.
\end{theorem}


\paragraph*{Structural Parameterizations.}
In scenarios where the size of the solution is large, it might be desired to drop the constraint on the size of the solution altogether, and seek to
parameterize the problem with some structural parameter of the input. In this setting,
we first consider the problem parameterized by the number of leaves of the tree and then extend this result to a more general parameter that takes into account the number of requests (terminal pair paths) passing through a vertex.
Technically,
we solve a different problem in this setting,
where we only have a uni-objective function
seeking to minimize the weight of the solution
(in contrast to the bi-objective function in the case of \bimctrees).
This problem is formally defined below.


\defproblem{\wmctreeslong\ (\wmctrees)}
{A tree $T$,
a collection of terminal pairs $\mathcal{P} \subseteq V(T) \times V(T)$,
a vertex weight function $\wt : V(T) \to \mathbb{N}$
and a positive integer $\ww$.}
{Does there exist $S \subseteq V(T)$
such that $\wt(S) \leq \ww$ and
$S$ intersects the unique $(s,t)$-path in $T$, for each $(s, t) \in \calP$?}


\wmc\ is another generalization of the vertex deletion variant of {\sc Multicut}.
The former problem has been studied on trees in the parameterized complexity setting
with respect to certain structural parameters.
In particular, Guo et al.~\cite[Theorem~9]{DBLP:journals/jda/GuoN06}
showed that \wmctrees\ is \FPT\ when the parameter is
the maximum number of $(s,t)$-paths
that pass through any vertex of the input.
We call this parameter the {\em request degree} $d$ of an instance.
Guo et al.~\cite{DBLP:journals/jda/GuoN06} gave an algorithm for
\wmctrees\ that runs in time $\OO(3^d \cdot n)$.

We first study \wmctrees\ when the parameter is the number of leaves of the tree.
The problem is polynomial time solvable on paths (\cref{lem:wt-mc-path})
but becomes \nph\ on (general) trees.
Thus, the number of leaves appears to be a natural parameter
which could explain the contrast between the above two results.
Formally, we prove the following theorem.

\begin{theorem}
  \label{thm:wt-mc-trees-l}
  \wmctrees can be solved in \timewmctrees\ time, where $\ell$ is the number of leaves in the input tree.
\end{theorem}


At the core of the algorithm for Theorem~\ref{thm:wt-mc-trees-l},
we again solve instances of \wdpc\ on trees,
but, in  this case, these instances
have a special structure:
they are subdivided stars
(i.e.\ trees with at most one vertex of degree at least $3$).
We show that these instances do not require the use of the flow augmentation technique.
In fact, these instances correspond to
the arcless instances of \wdpc\ in~\cite[Section~$6.2.2$]{kim2021directed} defined roughly as follows:
the input graph comprises of two designated vertices $s,t$
with internally vertex-disjoint paths from $s$ to $t$,
and the solution picks exactly one arc
from each of these internally vertex-disjoint paths.
Since the arcless instances can be solved faster than the general instances of \wdpc\
and do not require the usage of the flow augmentation technique~\cite[Lemma~6.12]{kim2021directed},
the algorithm for \wmctrees\ is deterministic and has a better running time.

As a final result,
we use the algorithm of Theorem~\ref{thm:wt-mc-trees-l}
as a subroutine to give an \FPT\ algorithm for \wmctrees\ %
that generalizes the result of Guo et al.~\cite[Theorem~9]{DBLP:journals/jda/GuoN06}
and Theorem~\ref{thm:wt-mc-trees-l}.
To do so, we define a new parameter that comprises both the request degree
and the number of leaves of the input instance.
An instance $(T,\mathcal{P},\wt,\ww)$ is {\em $(d,q)$-light}
if the following hold.
Let $Y$ be the set of vertices through which at most $d$ terminal pair paths of $\mathcal{P}$ pass.
Such vertices are called \emph{$d$-light vertices}.
Then for each connected component $C$ of $T-Y$,
the number of leaves of $T[N[C]]$ must be at most $q$
(see \cref{fig:d-ell-light-example} for an illustration of the definition).
We observe in \cref{sec:appendix-extras}
that it is crucial to consider the \emph{neighborhood} of the component, as
the problem is otherwise already \nph for $d=3$ and $q=2$.
We design a
dynamic programming
algorithm that stores partial solutions for every $d$-light vertex
using the algorithm of Theorem~\ref{thm:wt-mc-trees-l} as a subroutine to solve the problem on $(d,q)$-light instances.



\begin{theorem}
  \label{thm:wt-mc-trees-d-l}
  \wmctrees can be solved in \timewmctreesdl\ time on $(d,q)$-light instances.
\end{theorem}

Observe that an instance with
a tree on $\ell$ leaves
is a $(0,\ell)$-light instance,
and an instance with
the
request degree at most $d$
is a $(d,0)$-light instance.
Thus, \cref{thm:wt-mc-trees-d-l} implies
\cref{thm:wt-mc-trees-l} and~Theorem~9 in~\cite{DBLP:journals/jda/GuoN06},
up to the polynomial factors in the running time.




\paragraph*{\textbf{Our Methods.}}
Our algorithms for Theorems~\ref{thm:bi-mc-trees-k} and~\ref{thm:wt-mc-trees-l},
are crucially based on the observation mentioned earlier:
if every terminal path goes through a root $\rr$,
then the problem reduces to \wdpc.
In the vertex deletion version,
the \emph{vertex} $\mathcal{P}$-multicut in a tree
can be found using the algorithm for \wdpc,
by assigning the weight of a vertex to the unique edge
connecting it to its parent in $T$.
The general idea for both our algorithms is to design a branching algorithm
that effectively solves instances of the above-mentioned type
to reduce the measure in each branch.
Let $T$ be a rooted tree.
The goal is to
identify two vertices $x,y \in V(T)$ where $x$ is a descendant of $y$,
and branch on the possibility of a hypothetical solution
intersecting the $(y,x)$-path. 
If the solution does not intersect the $(y,x)$-path,
then contracting the edges of the $(y,x)$-path and making the resulting vertex undeletable, is a safe operation.
If the solution intersects the $(y,x)$-path,
then for each vertex $v$ on the $(y,x)$-path,
we increase the weight of $v$ by adding to it
the minimum weight of a solution in $T_v -\{v\}$ (where $T_v$ is the subtree of $T$ rooted at $V$),
and then forget about the terminal pair paths in $T_v - \{v\}$.
To update the weight of $v$,
one therefore needs to find a minimum weight solution in $T_v -\{v\}$.
For this reason, we choose the vertices $x,y$ so that
the instance restricted to $T_v - \{v\}$ can be solved using the algorithm for \wdpc.

If $x,y$ are vertices of degree at least $3$ (branching vertices) in $T$,
then contracting the $(y,x)$-path decreases the number of branching vertices in the resulting instance.
This choice of $x,y$ allows to design a branching algorithm
where the measure is the number of branching vertices,
and thus the number of leaves (Theorem~\ref{thm:wt-mc-trees-l}).
If $x,y$ are vertices of a minimum-size (unweighted) $\mathcal{P}$-multicut (which can be found in polynomial time),
then contracting the $(y,x)$-path 
decreases the size of a $\mathcal{P}$-multicut in the resulting instance.
This choice of $x,y$ allows the design of a branching algorithm
parameterized by the solution size (Theorem~\ref{thm:bi-mc-trees-k}).
Additionally,
if we choose $x$ to be the furthest branching vertex in $T$ (resp.~furthest vertex of $X$) from the root
and $y$ to be its unique closest ancestor that is a branching vertex (resp.~in $X$),
then for each vertex $v$ on the $(y,x)$-path,
the instance restricted on $T_v - \{v\}$ can indeed be solved using the algorithm for \wdpc.

\paragraph*{\textbf{Organization.}}
In Section~\ref{sec:preliminariesfull} we define some basic notation.
In Section~\ref{sec:bi-mc-k} we prove Theorem~\ref{thm:bi-mc-trees-k},
in Section~\ref{sec:param-l} we prove Theorem~\ref{thm:wt-mc-trees-l}, and
in Section~\ref{sec:combines-dl} we prove Theorem~\ref{thm:wt-mc-trees-d-l}.
We finally conclude in Section~\ref{sec:conclusion}.


\section{Notation and Preliminaries}
\label{sec:preliminariesfull}

For a positive integer $n$, we denote the set $\{1, 2, \dots, n\}$ by $[n]$.
We use $\mathbb{N}$ to denote the set of all non-negative integers.
Given a function $f : X \to Z$ and $Y \subseteq X$,
$f|_{Y}$ denotes the function $f$ restricted to $Y$.

\paragraph*{Graph Theory.}
For a (di-)graph $G$, $V(G)$ and $E(G)$
denote the set of vertices and edges (arcs) of $G$, respectively.
For an undirected graph $G$ and any $v \in V(G)$,
$N_G(v) = \{u: (u,v) \in E(G)\}$ denotes the (open) neighbourhood of $v$,
and $N_G[v]=N_G (v) \cup \{v\}$ denotes the closed neighbourhood of $v$.
The degree of $v$ is $|N_G(v)|$.
For a digraph $G$ and $v \in V(G)$,
the {\em in-degree} of $v$ is the number of vertices $u$ such that $(u,v) \in E(G)$.
When the graph $G$ is clear from the context we omit subscript $G$.
For any $u,v \in V(G)$,
a $(u,v)$-path in $G$ denotes a path from $u$ to $v$ in $G$.
For any $S \subseteq V(G)$, $G[S]$ denotes the graph $G$ induced on $S$.
We say $G' \subseteq G$ if $G'$ is a subgraph of $G$.
For any further notation from basic graph theory, we refer the reader to~\cite{diestel2005graph}.


\begin{figure}[t]
	\begin{center}
	\begin{tikzpicture}[scale=0.5]
		\draw[thick]
			(0, 0) node[anchor=south] {$\rr$} --
			(-2, -1) node[anchor=east] {$u$} --
			(-3, -2) --
			(-3, -3) --
			(-4, -4);

		\draw[thick]
			(0, 0) --
			(2, -1) node[anchor=west] {$v$} --
			(3, -2) --
			(3, -3) node[anchor=west] {$x$}--
			(4, -4);

		\draw[thick]
			(2, -1) -- (1, -2) -- (1, -3);

		\draw[draw=black, thick, dotted] (-1.5, -0.5) rectangle (-4.5,-4.75);
		\draw[draw=black, thick, dotted] (-1.75, -1.5) rectangle (-4.25,-4.5);

		\draw (-4,-1) node {$T_u$};
		\draw (-3.75,-2) node {$T^{\dag}_u$};

		\draw (4,-1) node {$T_{v,x}$};
		\draw (3.75,-2) node {$T^{\dag}_{v,x}$};

		\draw[draw=black, thick, dotted] (1.5, -0.5) rectangle (4.65,-4.75);
		\draw[draw=black, thick, dotted] (1.75, -1.5) rectangle (4.4,-4.5);

		\draw[thick]
			(3, -3) -- (2, -4);

		\draw[thick]
			(-3, -3) -- (-2, -4);

		\filldraw
			[black] (0, 0) circle [radius=3pt]
			[black] (-2, -1) circle [radius=3pt]
			[black] (-3, -2) circle [radius=3pt]
			[black] (-3, -3) circle [radius=3pt]
			[black] (-4, -4) circle [radius=3pt]
			[black] (2, -4) circle [radius=3pt]
			[black] (-2, -4) circle [radius=3pt]
			[black] (1, -2) circle [radius=3pt]
			[black] (1, -3) circle [radius=3pt];

		\filldraw
			[black] (2, -1) circle [radius=3pt]
			[black] (3, -2) circle [radius=3pt]
			[black] (3, -3) circle [radius=3pt]
			[black] (4, -4) circle [radius=3pt];

	\end{tikzpicture}
	\end{center}
	\caption{
	Tree rooted at $\rr$ and an illustration of $T_u, T^{\dag}_u, T_{v,x}, T^{\dag}_{v,x}$.
	\label{fig:rooted-and-dagged-tree}}
\end{figure}
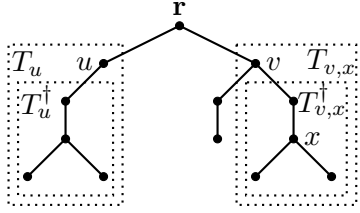

\paragraph*{Contraction.}
For any edge $e =(u,v)$ of $G$,
$G/e$ represents the graph $G$ obtained after contracting the edge $e$,
where the contraction of $e$ is defined as follows:
delete $u,v$ from $G$ and add a new vertex say $x_{uv}$
that is adjacent to the all the vertices
in $N[u]\cup N[v]\setminus\{u,v\}$,
that is all vertices
that were either adjacent to $u$ or $v$ or both.
We sometime also say that the edge $e$ is contracted onto the vertex $x_{uv}$.
For any $F \subseteq E(G)$,
$G/F$ denotes the graph obtained by contracting the edges of $F$
(one after the other in no specified order).
Formally speaking $G/F$ corresponds to a unique map
$\psi : V(G) \to V(G\setminus F)$,
such that
for each $u \in V(G / F), G[\psi^{-1}(u)]$ is connected.
Let $\mathcal{P} \subseteq V(G) \times V(G)$ be a set of pairs of vertices,
then $\mathcal{P}/F$ is obtained from $\mathcal{P}$
by replacing each $(u,v) \in \mathcal{P}$ by $(\psi(u),\psi(v))$.
Given $\mathcal{P} \subseteq V(G) \times V(G)$
and an induced subgraph $G'$ of $G$,
$\mathcal{P}|_{G'} \subseteq \mathcal{P}$ such that if $(u,v) \in \mathcal{P}$
then $(u,v) \in \mathcal{P}|_{G'}$ if and only if $u,v \in V(G')$.

\paragraph*{Terminal Pairs and Terminal Pair Paths.}
Recall that the instances of all our problems contain a set
$\mathcal{P} \subseteq V(T) \times V(T)$ where $T$ is a tree.
We interchangeably refer to a pair $ (s,t) \in \mathcal{P}$
as the terminal pair $(s,t)$ and as the unique $(s,t)$-path between in $T$.

\paragraph*{Trees.}
A {\em tree} $T$ is a connected acyclic graph.
A {\em subdivided star} is a tree with at most one vertex of degree at least $3$
(in other words, it is a tree obtained by repeatedly sub-dividing the edges of a star graph).
For any $x,y \in V(T)$, $P_{x,y}$ denotes the unique $(x,y)$-path in $T$.
Note that the vertices of such a path $P_{x,y}$ are ordered
starting from $x$ thus they have a natural order defined on them.
Let $T$ be a tree rooted at a vertex $\rr$.
A vertex $u \in V(T)$ is called a {\em furthest} vertex
from $v$
if $\textsf{dist}_T(u,v) = \max_{x \in V(T)}\{\textsf{dist}_T(x,v)\}$.
Here, $\textsf{dist}_T(u,v)$ denotes the length of the shortest $(u,v)$-path in $T$.
Similarly, $u \in V(T) \setminus \{\rr\}$ is closest to $v$,
if $0< \textsf{dist}_T(u,v) = \min_{x \in V(T)}\{\textsf{dist}_T(x,v)\}$.
We say that a vertex is {\em furthest} (resp.~{\em closest}) if it is furthest (resp.~closest) from $\rr$.
The sets $V_{\geq 3}(T)$ and $V_{=1}(T)$
denote the set of vertices of degree at least $3$,
and of degree equal to $1$, respectively.
The set $V_{\geq 3}(T)$ is also called the set of {\em branching vertices} of $T$
and the set $V_{=1}(T)$ is called the set of {\em leaves} of $T$.
Note that $|V_{\geq 3}(T)| \le |V_{=1}(T)|-1$.

The following notation (see Figure~\ref{fig:rooted-and-dagged-tree})
comes handy while describing the algorithms on a tree $T$ with root $\rr$.
Given $u,v \in V(T)$, $u$ is a {\em descendant} of $v$ in $T$,
if $v$ lies on the unique $(\rr,u)$ path in $T$
and $u$ is called an {\em ancestor} of $v$
if $u$ lies on the unique $(\rr,v)$-path in $T$ ($u$ could be equal to $v$).
We denote by $T_u$ the subtree of $T$ rooted at $u$
and $T^{\dag}_u = T_u \setminus \{u\}$.
For any descendant $x$ of $u$,
the tree denoted by $T_{u,x}$ is defined as follows.
Let $ \{v_1, \ldots,v_p\}$ be the children of $u$ in $T$
and say $x \in V(T_{v_i})$.
Then $T_{u,x} = T_u \setminus (\cup_{j \in [p] \setminus \{i\}} T_{v_j})$.
Observe that $T_{u,x}$ is connected.
Also define $T^{\dag}_{u,x} = T_{u,x} \setminus \{u\}$.

Let $X \subseteq V(T)$,
then the \texttt{lca-closure} of $X$ is the set $X'$ obtained from $X$
by repeatedly adding, for each pair of vertices $u,v \in X'$,
the least common ancestor $w$ to $X'$,
that is the vertex $w \in V(T)$ furthest from $\rr$
such that $u,v \in V(T_w)$.
Note that $|X'| \leq 2 |X|$ (because $|V_{\geq 3}(T)| < |V_{=1}|(T)$).
We say that a set $X$ is {\em closed under taking \texttt{lca}}
if for every pair of vertices of $X$, their least common ancestor is in $X$.
Whenever we talk about a rooted (undirected) tree $T$ in a directed setting,
we refer to the tree $T$
where each vertex except the root has in-degree exactly one.

\paragraph*{Parameterized Complexity.}
The input of a parameterized problem comprises of an instance $I$,
which is an input of the classical instance of the problem,
and an integer $k$, which is called the parameter.
A problem $\Pi$ is said to be \emph{fixed-parameter tractable} or \FPT\ %
if given an instance $(I,k)$ of $\Pi$,
we can decide whether $(I,k)$ is a \yes\ instance of $\Pi$
in time $f(k)\cdot |I|^{\OO(1)}$.
Here, $f(\cdot)$ is some computable function whose value depends only on $k$.
We say that two instances, $(I, k)$ and $(I', k')$,
of a parameterized problem $\Pi$ are \emph{equivalent}
if $(I, k) \in \Pi$ if and only if $(I', k') \in \Pi$.
For more details on parameterized algorithms, and in particular parameterized branching algorithms,
we refer the reader to the book by Cygan et al.~\cite{DBLP:books/sp/CyganFKLMPPS15}.

\section{\bimctrees\ Parameterized by the Solution Size}\label{sec:bi-mc-k}
In this section, we prove Theorem~\ref{thm:bi-mc-trees-k} by designing a branching algorithm.
In order to reduce the measure of a given instance,
 our branching algorithm requires a solution for the instances
where every terminal pair path passes through a single vertex.
Let $\mathcal{I}=(T,\mathcal{P},\rr,\wt,\kk)$ be an instance such that
all the terminal pair paths of $\mathcal{P}$ pass through $\rr$,
and $\wt : V(T) \to \mathbb{N}$ is a vertex weight function.
Let $\overrightarrow{T}$ be the directed tree obtained by
orienting the edges of $T$ so that all the vertices, except for $\rr$,
have in-degree exactly one, while $\rr$ has in-degree zero.
In other words, the oriented tree  $\overrightarrow{T}$ is an out-tree rooted at $\rr$.
We define an edge weight function $\wt' : E(\overrightarrow{T}) \to \mathbb{N}$
such that for every arc $e = (u,v) \in E(\overrightarrow{T})$, $\wt'(e) = \wt(v)$.
Then it can be easily seen that
$Z \subseteq E(\overrightarrow{T})$ is a $\mathcal{P}$-dpc in $T$ with
$\wt'(Z) = \ww$ if and only if $S=\{v : (u,v) \in Z\} \subseteq V(T) \setminus \{\rr\}$
(that is, $S$ is obtained from $Z$ by picking the heads of all the arcs in $Z$)
is a $\mathcal{P}$-multicut in $T$ with $\wt(S) = \ww$
(see Figure~\ref{fig:example-wmc-to-dpc}).
Let $\Adpc$ be the algorithm that takes as input an instance $\mathcal{I}$ as above,
constructs the edge-weight function $\wt'$ and uses the \wdpc\ algorithm
of Kim et al.~\cite[Section~6.1ff]{kim2021directed}
to solve the instance $(\overrightarrow{T}, \mathcal{P},\rr,\wt',\kk)$.   This runs in randomized $2^{\OO(\kk^{4})}\cdot n^{\OO(1)}$-time
\footnote{The dependency in $\kk$
is not explicit in \cite{kim2021directed} but can be easily deduced.}.
Therefore,  $\Adpc$ outputs the minimum \emph{weight}
of a solution of $\mathcal{I}$ if it exists,
and $\infty$ otherwise.
In particular, if $\mathcal{P} = \emptyset$ then $\Adpc$ outputs $0$.

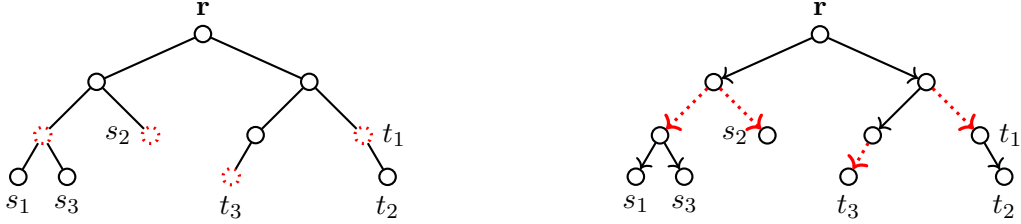
\begin{figure}[t]
\tikzset{node distance=1em and 0.4em}

\begin{subfigure}{.49\textwidth}
\centering
\begin{tikzpicture}
  \node[stan, label=$\rr$] at (0,0) (n1) {};
  \node[left = 3em of n1] (h1) {};
  \node[right= 3em of n1] (h2) {};

  \node[stan, below = of h1] (n2) {};
  \node[stan, below = of h2] (n3) {};

  \node[stan, del, very thick, below left = 2em of n2] (n4) {};
  \node[stan, del, very thick, label=left:$s_2$, below right= 2em of n2](n5){};
  \node[stan, below left = 2em of n3] (n6) {};
  \node[stan, del, very thick, label=right:$t_1$, below right= 2em of n3](n7){};

  \node[stan, label=below:$s_1$, below left = of n4] (n8) {};
  \node[stan, label=below:$s_3$, below right= of n4] (n9) {};
  \node[stan, del, very thick, label=below:$t_3$, below left = of n6] (n12) {};
  \node[stan, label=below:$t_2$, below right= of n7] (n15) {};

  \draw[-,thick] (n1) -- (n2);
  \draw[-,thick] (n1) -- (n3);
  \draw[-,thick] (n2) -- (n4);
  \draw[-,thick] (n2) -- (n5);
  \draw[-,thick] (n3) -- (n6);
  \draw[-,thick] (n3) -- (n7);
  \draw[-,thick] (n4) -- (n8);
  \draw[-,thick] (n4) -- (n9);
  \draw[-,thick] (n6) -- (n12);
  \draw[-,thick] (n7) -- (n15);
\end{tikzpicture}
\caption{
The instance for \bimctrees with a marked solution.
}
\end{subfigure}%
\hfill
\begin{subfigure}{.49\textwidth}
\centering
\begin{tikzpicture}
  \node[stan, label=$\rr$] at (0,0) (n1) {};
  \node[left = 3em of n1] (h1) {};
  \node[right= 3em of n1] (h2) {};

  \node[stan, below = of h1] (n2) {};
  \node[stan, below = of h2] (n3) {};

  \node[stan, below left = 2em of n2] (n4) {};
  \node[stan, label=left:$s_2$, below right= 2em of n2] (n5) {};
  \node[stan, below left = 2em of n3] (n6) {};
  \node[stan, label=right:$t_1$, below right= 2em of n3] (n7) {};

  \node[stan, label=below:$s_1$, below left = of n4] (n8) {};
  \node[stan, label=below:$s_3$, below right= of n4] (n9) {};
  \node[stan, label=below:$t_3$, below left = of n6] (n12) {};
  \node[stan, label=below:$t_2$, below right= of n7] (n15) {};

  \draw[->,thick] (n1) -- (n2);
  \draw[->,thick] (n1) -- (n3);
  \draw[->,very thick,del] (n2) -- (n4);
  \draw[->,very thick,del] (n2) -- (n5);
  \draw[->,thick] (n3) -- (n6);
  \draw[->,very thick,del] (n3) -- (n7);
  \draw[->,thick] (n4) -- (n8);
  \draw[->,thick] (n4) -- (n9);
  \draw[->,very thick,del] (n6) -- (n12);
  \draw[->,thick] (n7) -- (n15);
\end{tikzpicture}
\caption{
The instance for \wdpc where the corresponding solution is marked.
}
\end{subfigure}

\caption{
One-to-one correspondence between the solutions of \bimctrees\ and the solutions of  \wdpc when all the paths of $\mathcal{P}$ pass through $\rr$.
}
\label{fig:example-wmc-to-dpc}
\end{figure}



\paragraph*{Branching Algorithm.}
Let $(T,\mathcal{P},\wt,\ww,\kk)$ be an instance of \bimctrees.
Fix an arbitrary vertex $\rr \in V(T)$ to be the root of $T$.
We begin by finding a set $X \subseteq V(T)$ which is a $\mathcal{P}$-multicut in $T$ and is closed under taking \texttt{lca} (least common ancestor).
To find $X$, we first compute a \emph{unweighted} $\calP$-multicut
$X_{\text{opt}} \subseteq V(T)$ in $T$ of minimum size.
The set $X_{\text{opt}}$ can be found in polynomial time (folklore)
by the following greedy algorithm. 
Initialize $X_{\text{opt}} = \emptyset, T'=T$,  and $\mathcal{P}'=\mathcal{P}$.
Let $v \in V(T')$
be a furthest vertex from $\rr$ such that
there exists $(s,t) \in \mathcal{P}'$ with $s,t \in V(T_v)$.
By the choice of $v$,
the $(s,t)$-path
(and every terminal pair path in $\calP|_{T_v}$)
passes through $v$.
It is easy to see that
there is a minimum-size $\mathcal{P}'$-multicut containing $v$.
Set $X_{\text{opt}}=X_{\text{opt}} \cup \{v\}$,
$\mathcal{P}'=\mathcal{P}' \setminus \calP|_{T_v}$,
$T'= T' \setminus T_v$,
and repeat the procedure until $\mathcal{P}'=\emptyset$.
At the end of the procedure,
$X_{\text{opt}}$ is a minimum-size $\mathcal{P}$-multicut in $T$.
If $|X_{\text{opt}}| > \kk$, report \no.
Otherwise, let $X$ be the \texttt{lca-closure} of
$X_{\text{opt}}$ in $T$.
Hence, $|X| \leq 2 \kk$.

A notable property of a $\calP$-multicut $X$ closed under taking \texttt{lca} is that for any $x \in X$,
if $y \in X$ is the unique closest ancestor of $x$ in $T$,
then for each $v \in V(P_{y,x})\setminus\{y\}$,
all the terminal pair paths of $\mathcal{P}|_{T_v}$
either pass through $x$, or are contained in $T_x$.
Indeed, if $T_v \setminus T_x$ contains a path of $\mathcal{P}$,
then any $\mathcal{P}$-multicut intersects $V(T_v \setminus T_x)$.
Then there exists a vertex $y' \in X$ such that $y'\neq x$ lies on $P_{v,x} \subseteq P_{y,x}^\dag$, contradicting the choice of $y$.

We design a branching algorithm
whose input is $\mathcal{I} = (T,\calP,\wt,\ww,\kk,X)$
where $X$ is $\calP$-multicut $X \subseteq V(T)$ closed under taking \texttt{lca},
and where the measure of an instance $\mathcal{I}$ is defined as
$\mu(\mathcal{I})= |X|$. Note that, as mentioned above, $\mu(\mathcal{I}) \leq 2\kk$.
The base case of the branching algorithm occurs
in the following scenarios.
\begin{enumerate}
  \item If $\mu(\mathcal{I})=0$, then $\emptyset$ is a solution of $\mathcal{I}$.
  Return \yes\
  iff
  $\kk \geq 0$ and
  $\ww \geq 0$.
  \item If $\mu(\mathcal{I})=1$,  let $X=\{x\}$.
  In this case, since all the paths of $\mathcal{P}$
  pass through $x$,
  return \yes\ if and only if
  the $\Adpc(T,\mathcal{P},x,\wt,\kk) \leq \ww$.
  \item If $k < 0$, or $k \leq 0$ and $\mathcal{P} \neq \emptyset$, then return \no.
\end{enumerate}


If $\mu(\mathcal{I}) \geq 2$ (that is, $|X| \geq 2$),
then let $x \in X$ be a furthest vertex from $\rr$
and let $y \in X$ be its unique closest ancestor.
We branch in the following two cases.\\

\noindent
\textbf{Case 1.} \emph{There exists a solution of $\mathcal{I}$ that does not intersect $V(P_{y,x})$.}
In this case, we return the instance $\mathcal{I}_1=(T_1,\allowbreak \mathcal{P}_1, \allowbreak \wt_1,\ww ,\kk,X_1)$
where $T_1 = T / E(P_{y,x})$, $\mathcal{P}_1=\mathcal{P}/E(P_{y,x})$.
Let the vertex onto which the edges of $P_{y,x}$
are contracted be $y^{\circ}$.
%
Then $\wt_1(y^{\circ})=\ww+1$ and,
for each  $v \in V(T_1) \setminus \{y^{\circ}\}$, $\wt_1(v) =\wt(v)$. 
Observe that $(X \setminus \{x,y\}) \cup \{y^{\circ}\}$ is a $\mathcal{P}_1$-multicut in $T_1$
and is closed under taking \texttt{lca} and thus, we may set $X_1 = (X \setminus \{x,y\}) \cup \{y^{\circ}\}$.
Clearly, $\mu(\mathcal{I}_1) < \mu(\mathcal{I})$ and $\mathcal{I}_1$ can be constructed in polynomial time.\\


\noindent
\textbf{Case 2.} \emph{There exists a solution of $\mathcal{I}$ that intersects $V(P_{y,x})$.} 
In this case, the idea is the following:
for each vertex $v$ on the $P_{y,x}$ path,
we increase the weight of $v$
by the weight of the solution in the tree $T^{\dag}_{v,x}$
(the tree strictly below $v$).
To do so,
the size of a solution in the tree $T^\dag_{v,x}$ is first guessed.
Once the weights are updated, 
we can forget the terminal pairs
contained in the tree $T^{\dag}_{y,x}$
and just remember that the solution picks a vertex from $P_{y,x}$.
This is formalized below.

Let $S$ be a solution which intersects $V(P_{y,x})$ and let $z \in V(P_{y,x})$ be the vertex in $S$ closest to $y$.
Then we further branch into $\kk+1$ branches
where each branch corresponds to the guess on $|S \cap T^{\dag}_{z,x}|$.
More precisely, for every $i \in \{0\} \cup [k]$, we create the instance
$\mathcal{I}_{2,i}=(T_2, \mathcal{P}_2,\wt_{2,i}, \ww, \kk-i,X_2)$
where $T_2 = T \setminus T^\dag_x$,
$\mathcal{P}_2=\mathcal{P}|_{T_2} \setminus (V(T^{\dag}_{y,x}) \times  V(T^{\dag}_{y,x})) \cup \{(y,x)\}$ and
$\wt_{2,i}$ is defined below (see Figure~\ref{fig:wt-update-fpt-k}).
\[
  \wt_{2,i}(v) =
  \begin{cases}
      \wt(v) + \Adpc(T^{\dag}_{v,x}, \mathcal{P}|_{T^{\dag}_{v,x}}, x, \wt|_{V(T^{\dag}_{v,x})},i) & \text{ if } v \in V(P_{y,x})\\
  \wt(v) & \text{ otherwise.}
  \end{cases}
\]
%
Observe that the set $X \setminus \{x\}$ is a $\mathcal{P}_2$-multicut in $T_2$
with $y \in X \setminus \{x\}$.
The only paths that might not be cut are the ones in $\mathcal{P}|_{T^{\dag}_{y,x}}$
as they pass through $x$,
but they are not contained in $\mathcal{P}_2$ by definition.
Also $X \setminus \{x\}$ is closed under taking \texttt{lca} in $T_2$, thus
we may set $X_2 = X \setminus \{x\}$.
Clearly, $\mu(\mathcal{I}_{2,i}) < \mu(\mathcal{I})$ for each $i \in \{0\} \cup [k]$.

\begin{figure}[t]

\begin{subfigure}[t]{.49\textwidth}
  \centering
  \begin{tikzpicture}
    \node[] at (0,0) (root) {};
    \node[stan,label=left:$t_3$, below left = of root] (t3) {};
    \node[stan,contract,label=left:$\bm y$,below right = of t3] (y) {};
      \node[stan,contract,label=left:$ $,below left = of y] (y_kid) {};
        \node[stan,contract,label=left:$s_4$,below left= of y_kid](s4){};
          \node[stan,contract,label=left:$\bm x$,below left = of s4] (x) {};
            \node[stan,label=left:$ $,below left =  of x] (x_kid) {};
              \node[stan,label=below:$s_1$,below left = of x_kid] (s1) {};
              \node[stan,label=below:$s_2$,below right= of x_kid] (s2){};
            \node[stan,label=below:$t_1$,below right = of x] (t1) {};
        \node[stan,label=left:$ $,below right = of y_kid] (y_kid_kid) {};
          \node[stan,label=below:$t_2$,below left = of y_kid_kid] (t2) {};
          \node[stan,label=below:$s_3$,below right = of y_kid_kid] (s3) {};
      \node[stan,label=right:$t_4$,below right = of y] (t4){};


    \draw[-,thick, dashed] (root) -- (t3);
    \draw[-,thick, dashed] (t3) -- (y);
      \draw[-,thick, dashed] (y) -- (t4);
      \draw[-,thick,contract] (y) -- (y_kid);
        \draw[-,thick,contract] (y_kid) -- (s4);
          \draw[-,thick,contract] (s4) -- (x);
            \draw[-,thick] (x) -- (x_kid);
              \draw[-,thick] (x_kid) -- (s1);
              \draw[-,thick] (x_kid) -- (s2);
            \draw[-,thick] (x) -- (t1);
        \draw[-,thick] (y_kid) -- (y_kid_kid);
          \draw[-,thick] (y_kid_kid) -- (t2);
          \draw[-,thick] (y_kid_kid) -- (s3);
  \end{tikzpicture}
  \caption{
  Case~1.
  The marked edges are contracted onto the undeletable $y^\circ$.
  }
  \label{fig:wt-update-fpt-k:case1}
\end{subfigure}%
\hfill
\begin{subfigure}[t]{.49\textwidth}
  \centering
  \begin{tikzpicture}
    \node[] at (0,0) (root) {};
    \node[stan,label=left:$t_3$, below left = of root] (t3) {};
    \node[stan,changed,label=left:$\bm y$,below right = of t3] (y) {};
      \node[stan,changed,label=left:$ $,below left = of y] (y_kid) {};
        \node[stan,changed,label=left:$s_4$,below left= of y_kid](s4){};
          \node[stan,changed,label=left:$\bm x$,below left = of s4] (x) {};
            \node[stan,del,label=left:$ $,below left =  of x] (x_kid) {};
              \node[stan,del,label=below:$s_1$,below left=of x_kid](s1){};
              \node[stan,del,label=below:$s_2$,below right=of x_kid](s2){};
            \node[stan,del,label=below:$t_1$,below right = of x] (t1) {};
        \node[stan,label=left:$ $,below right = of y_kid] (y_kid_kid) {};
          \node[stan,label=below:$t_2$,below left = of y_kid_kid] (t2) {};
          \node[stan,label=below:$s_3$,below right = of y_kid_kid] (s3) {};
      \node[stan,label=right:$t_4$,below right = of y] (t4){};


    \draw[-,thick, dashed] (root) -- (t3);
    \draw[-,thick, dashed] (t3) -- (y);
      \draw[-,thick, dashed] (y) -- (t4);
      \draw[-,thick] (y) -- (y_kid);
        \draw[-,thick] (y_kid) -- (s4);
          \draw[-,thick] (s4) -- (x);
            \draw[-,thick,del] (x) -- (x_kid);
              \draw[-,thick,del] (x_kid) -- (s1);
              \draw[-,thick,del] (x_kid) -- (s2);
            \draw[-,thick,del] (x) -- (t1);
        \draw[-,thick] (y_kid) -- (y_kid_kid);
          \draw[-,thick] (y_kid_kid) -- (t2);
          \draw[-,thick] (y_kid_kid) -- (s3);
  \end{tikzpicture}
  \caption{
  Case~3.
  $T^\dag_x$ is deleted (dotted part).
  The weight of the filled vertices includes the weight of the solution from below.
  }
  \label{fig:wt-update-fpt-k:case3}
\end{subfigure}

%

\caption{
The branches of the algorithm for Theorem~\ref{thm:bi-mc-trees-k}
with $(s_i,t_i) \in \Pairs$.
\label{fig:wt-update-fpt-k}}
\end{figure}
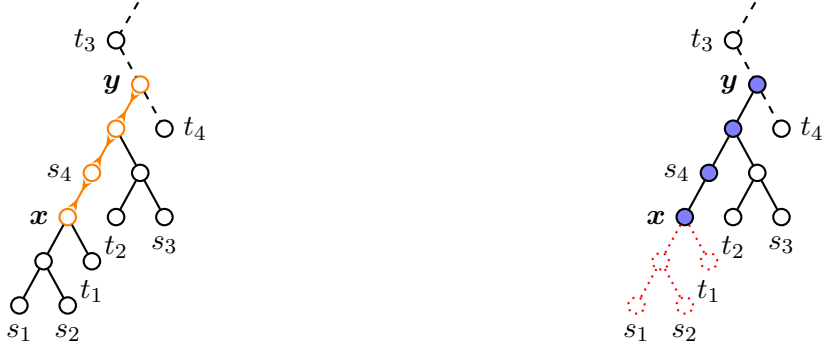

\begin{lemma}
  \label{lem:correctness:bi-mc-k}
  $\mathcal{I}$ is a \yes-instance if and only if at least one of
  $\mathcal{I}_1, \mathcal{I}_{2,0}, \ldots, \allowbreak \mathcal{I}_{2,\kk}$ is a \yes-instance.
\end{lemma}

\begin{proof}
  ($\Rightarrow$) Assume that $\mathcal{I}$ is a \yes-instance and let $S$ be a minimal solution of $\mathcal{I}$.
  Suppose first that $S \cap V(P_{y,x}) = \emptyset$ and
  consider a path $P_{s,t}$ of $\mathcal{P}_1$.
  Then $(s,t) \neq (y^\circ,y^\circ)$ for otherwise,
  $S$ would not intersect the path in $\mathcal{P}$ corresponding to the pair $(s,t) \in \mathcal{P}_1$.
  If $y^\circ \notin \{s,t\}$ then $(s,t) \in \mathcal{P}$ and so, $S$ intersects the path $P_{s,t}$.
  Otherwise, assume, without loss of generality, that $s =y^\circ$ and
  let $(z,t) \in \mathcal{P}$ where $z \in V(P_{y,x})$, be the terminal pair in $\mathcal{P}$ corresponding to $(s,t)$.
  Then, since $P_{z,t}$ is intersected by $S \setminus V(P_{y,x})$, $P_{s,t}$ is also intersected by $S \setminus \{y^\circ\}$.
  Thus, $S$ is a solution for $\mathcal{I}_1$.

  Suppose next that $S \cap V(P_{y,x}) \neq \emptyset$
  and let $z \in V(P_{y,x})$ be the vertex in $S$ closest to $y$.
  Observe that since $X$ is a $\mathcal{P}$-multicut in $T$ and $x \in X$ is a furthest vertex in $T$ from $\rr$,
  every path of $\mathcal{P}$ contained in $T_x$ passes through $x$.
  Similarly, if $z \neq x$ then, from the choice of $x$ and $y$,
  each terminal pair path contained in $T^{\dag}_{z,x}$ passes through $x$:
  indeed, if there exists a terminal pair path contained in $T^{\dag}_{z,x} \setminus T_x$,
  then it is not intersected by $X$,
  a contradiction to the fact that $X$ is a $\calP$-multicut.
  Let $S^* = S \cap T^{\dag}_{z,x}$ and let $i = |S^*|$.
  Note that if $z = x$ then $\calP_{T^{\dag}_{x,x}} = \emptyset$ by the above, and thus,
  $S^* = \emptyset$ by minimality of $S$.
  Since $S^*$ is a $\mathcal{P}|_{T^{\dag}_{z,x}}$-multicut, it follows that
  $\wt(S^*) \geq \Adpc({T^{\dag}_{z,x}}, \mathcal{P}|_{T^{\dag}_{z,x}},x,\wt|_{{T^{\dag}_{z,x}}},i)$.
  Now let $S' = S \setminus S^*$.
  Note that $z \in S'$; in fact, $S' \cap V(P_{y,x}) = \{z\}$ by the choice of $z$.
  We claim that $S'$ is a solution for $\mathcal{I}_{2,i}$.
  Clearly, $|S'| = |S| - |S^*| \leq \kk - i$.
  Furthermore, $\wt_{2,i}(S')=\wt(S) - \wt(S^*) - \wt(z) +\wt_{2,i}(z)$ and
  since $z \in V(P_{y,x})$,
  $\wt_{2,i}(z) \leq \wt(z) + \wt(S^*)$.
  Thus, $\wt_{2,i}(S') \leq \wt(S) \leq \ww$.
  We now show that $S'$ is a $\mathcal{P}_2$-multicut.
  Consider a path $P_{s,t}$ of $\mathcal{P}_2$.
  Since by construction, $\mathcal{P}_2 \cap (V(T_{y,x}^\dag) \times V(T_{y,x}^\dag)) = \emptyset$,
  at most one of $s$ and $t$ belongs to $V(T_{y,x}^\dag)$.
  Suppose first that $\{s,t\} \cap V(T_{y,x}^\dag) \neq \emptyset$, say $s \in V(T_{y,x}^\dag)$ without loss of generality.
  If $s \in V(P_{y,z}) \setminus \{y,z\}$ then,
  by the choice of $z$ and because $S$ is a $\calP$-multicut,
  $P_{s,t}$ is intersected by $S \setminus V(T_{y,x}^\dag) \subseteq S'$.
  Otherwise, $P_{s,t}$ passes through $z$ and is therefore intersected by $S'$.
  Since it is clear that $P_{s,t}$ is intersected by $S'$ if $\{s,t\} \cap V(T_{y,x}^\dag) = \emptyset$,
  we conclude that $S'$ is indeed a $\mathcal{P}_2$-multicut.

  ($\Leftarrow$) Suppose first that $\mathcal{I}_1$ is a \yes-instance and
  let $S_1$ be a solution of $\mathcal{I}_1$.
  Since $\wt_1(y^{\circ}) =\ww +1$,  $y^{\circ} \not \in S$.
  This implies, in particular, that $(y^\circ,y^\circ) \notin \calP_1$
  and thus, no path of $\calP$ is contained in $P_{y,x}$.
  Therefore, $S_1$ is a solution for $\mathcal{I}$.
  Suppose next that there exists $i \in \{0,\ldots,\kk\}$ such that $\mathcal{I}_{2,i}$ is a \yes-instance
  and let $S_{2,i}$ be a minimal solution of $\mathcal{I}_{2,i}$.
  We first claim that $|S_{2,i} \cap V(P_{y,x})| =1$.
  Indeed, observe that $S_{2,i} \cap V(P_{y,x}) \neq \emptyset$ since $(y,x) \in \mathcal{P}_2$.
  For the sake of contradiction,
  suppose that there exist $z,z' \in S_{2,i} \cap V(P_{y,x})$
  such that $z' \neq z$,
  say $z'$ is a descendant of $z$.
  Since, by construction, no path of $\mathcal{P}_2$ is contained in $T^\dag_{y,x}$,
  each path of $\mathcal{P}_2$ that passes through $z'$, also passes through $z$.
  Thus, $S_{2,i} \setminus \{z'\}$ is a $\mathcal{P}_2$-multicut, contradicting the minimality of $S_{2,i}$.
  Let $S_{2,i} \cap V(P_{y,x}) = \{z\}$.
  As argued above,
  if $z \neq x$, then,
  from the choice of $x$ and $y$,
  every path of $\calP$ contained in $T^\dag_{z,x}$
  passes through $x$.
  Similarly, every path of $\calP$ contained in $T_x$ passes through $x$.
  Let $S^*$ be a $\calP|_{T^{\dag}_{z,x}}$-multicut of size at most $i$ such that $\wt(S^*)$ is minimum.
  Then $\wt(S^*) = \Adpc(T^{\dag}_{z,x},\mathcal{P}_{T^{\dag}_{z,x}},x,\wt|_{T^{\dag}_{z,x}},i)$.
  Let $S = S_{2,i} \cup S^*$.
  We claim that $S$ is a solution for $\mathcal{I}$.
  Indeed, first note that $|S| = |S_{2,i}| + |S^*| \leq \kk - i + i = \kk$.
  Furthermore, since $S_{2,i} \cap V(P_{y,x}) = \{z\}$,
  $\wt(S_{2,i}) = \wt_{2,i}(S_{2,i}) - \wt_{2,i}(z) + \wt(z)$
  and  $\wt_{2,i}(z) = \wt(z) + \wt(S^*)$.
  Thus, $\wt(S) = \wt(S_{2,i}) + \wt(S^*) \leq \wt_{2,i}(S_{2,i}) \leq \ww$.
  We now show that $S$ is a $\calP$-multicut.
  Since $S_{2,i} \subseteq S$ and
  $S_{2,i}$ is a $\mathcal{P}_2$-multicut,
  any path of $\calP$ fully contained in $V(T) \setminus V(T_{y,x}^\dag)$ is intersected by $S$.
  Consider now a path $P_{s,t}$ of $\calP$ that intersects $V(T_{y,x})$
  If $P_{s,t}$ is fully contained in $T^\dag_{z,x}$,
  then it is intersected by $S^*$.
  Similarly, if $P_{s,t}$ passes through $z$,
  then it is intersected by $S$ since $z \in S$.
  If $P_{s,t}$ passes through $y$ without containing $z$,
  then $P_{s,t} \in \mathcal{P}_2$ and so,
  by the choice of $z$,
  $P_{s,t}$ is intersected by $S_{2,i} \setminus \{z\} \subseteq S$.
  Observe finally that $P_{s,t}$ is not fully contained in $V(P^\dag_{y,z}) \setminus \{z\}$
  for otherwise, $P_{s,t}$ is not intersected by $X$,
  a contradiction to the fact that $X$ is a $\calP$-multicut.
  Therefore, $S$ is a solution for $\mathcal{I}$.
\end{proof}


\begin{proof}
  [Proof of Theorem~\ref{thm:bi-mc-trees-k}]
  Let $\mathcal{I}=(T,\mathcal{P},\wt,\ww,\kk)$ be an instance of \bimctrees.
  Lemma~\ref{lem:correctness:bi-mc-k} shows that the above algorithm correctly
  solves the problem.
  The described algorithm does a $(\kk+2)$-way branching,
  where the measure of the input instance is bounded by $2 \kk$ and
  drops
  by at least $1$ in every branch.
  Since the branching stops when the measure is at most $1$,
  the total number of branching nodes of the algorithm is at most $(\kk+2)^{{2 \kk} +1}$.
  Since $\mathcal{I}_1$ can be constructed in polynomial time and each instance $\mathcal{I}_{2,i}$
  can be constructed by making $\OO(n)$ calls to $\Adpc$, 
  the final running time is \timedpc.
\end{proof}

\section{\wmctrees\ Parameterized by the Number of Leaves}\label{sec:param-l}
In this section, we prove Theorem~\ref{thm:wt-mc-trees-l}.
We first show that the problem on sub-divided stars can be solved
without using the flow augmentation from~\cite{kim2021directed} (Lemma~\ref{lem:wt-mc-sub-star}).
Towards this,
we first design a simple polynomial-time algorithm for the problem on paths (Lemma~\ref{lem:wt-mc-path})
and use it to eliminate the terminal pair paths that do not pass through the high degree vertex of the sub-divided star.
We then observe that the problem on sub-divided stars, when each terminal pair path pass through the high degree vertex, corresponds to the arcless instances of~\cite[Section~$6.2.2$]{kim2021directed}, which can be solved faster~\cite[Lemma~$6.12$]{kim2021directed} (Proposition~\ref{prop:verwdpc-arcless}).
We then use the algorithm of Lemma~\ref{lem:wt-mc-sub-star} as a subroutine to design
a branching algorithm that proves Theorem~\ref{thm:wt-mc-trees-l}.

\begin{lemma}\label{lem:wt-mc-path}
  Let $T$ be a disjoint union of paths,
  $\mathcal{P} \subseteq V(T) \times V(T)$ and
  $\wt: V(T) \to \mathbb{N}$.
  There is an algorithm $\Apath$ that outputs the weight of a $\mathcal{P}$-multicut $S \subseteq V(T)$,
  in $T$ such that $\wt(S)$ is minimum,
  in polynomial time.
\end{lemma}
\begin{proof}
  If $T$ is the union of at least two disjoint paths,
  then it is enough to solve the problem on each path independently and
  output the sum of the weights returned in each instance. Without loss of generality,
  let $T$ be a path $(v_1, \ldots, v_n)$.
  For each $i \in [n]$, let $T_i=T[\{v_1, \ldots,v_i\}]$.
  We use dynamic programming to compute for each $i \in [n]$, $\textsf{B}[i]$ which stores the minimum weight of a $\mathcal{P}|_{T_i}$-multicut in $T_i$.
  This is computed as follows.
  For any $i \in [n]$ such that $\calP|_{T_i} \neq \emptyset$,
  let $i^* \leq i$ be the largest index such that
  there exists $(s,t) \in \mathcal{P}$ where $P_{s,t} \subseteq T[\{v_j~|~j \in\{i^*,\ldots,i\}\}]$. Then $\textsf{B}[i]$
  is computed as follows.

  \[
    \textsf{B}[i] =
    \begin{cases}
      0 & \text{if } \mathcal{P}|_{T_i} = \emptyset\\
      \wt(v_1) & \text{if } i=1\\
      \min_{i^* \leq j \leq i} \{ \wt(v_j) + \textsf{B}[j-1]\} & \text{otherwise.}
    \end{cases}
  \]

  The algorithm then returns $\textsf{B}[n]$.
  Clearly, the algorithm runs in polynomial time.
  If $\mathcal{P}|_{T_i} = \emptyset$,
  then $S = \emptyset$ is a solution.
  Otherwise, if $i=1$, then $(v_1,v_1) \in \mathcal{P}$ and so,
  $S=\{v_1\}$ is a minimum weight solution.
  If $i >1$, then from the choice of $i^*$,
  for any minimal $\mathcal{P}|_{T_i}$-multicut $S$,
  $|S \cap \{v_{i^*}, \ldots, v_i\}| =1$.
  If $S \cap \{v_{i^*}, \ldots, v_i\}=\{v_j\}$,
  then by induction,
  $\wt(S) = \wt(v_j)  + \textsf{B}[j-1]$.
  Again from the choice of $i^*$,
  for any $j \in \{i^*, \ldots, i\}$,
  $v_j$ union any $\mathcal{P}|_{T_{j-1}}$-multicut
  is also a $\mathcal{P}|_{T_i}$-multicut.
\end{proof}

An $\rr$-rooted subdivided star is a subdivided star with root $\rr$,
where $\rr$ is a highest degree vertex of $T$, that is, $\rr$ is the unique degree $3$ vertex,
if it exists, or any arbitrary internal vertex if $T$ is a path.
Consider an instance $(T,\mathcal{P},\rr,\wt)$ such that
$T$ is an $\rr$-rooted subdivided star
and all the paths of $\mathcal{P}$ in $T$ pass through $\rr$.
The goal is to find a $\mathcal{P}$-multicut $S$ such that
$\wt(S)$ is minimum.
We show that such instances corresponds to the {\em arcless} instances
of \wdpc\ as defined in~\cite[Section~$6.2.2$]{kim2021directed}.
An instance $(T,\mathcal{P},\wt)$ is an arcless instance if
(i) given two designated vertices $s,t$, the graph $T$ consists of only internally vertex-disjoint paths from $s$ to $t$, and
(ii) if it is a \yes-instance, then there exists a solution for this instance that intersects every $(s,t)$-path exactly once.

The following result follows from~\cite[Section~$6.2.2$, Lemma~$6.12$]{kim2021directed}.
The root of a subdivided star, that is not a path, is the unique branching vertex. 
\begin{proposition}[\cite{kim2021directed}]
  \label{prop:verwdpc-arcless}
  Given an instance $(T,\mathcal{P} \subseteq V(T) \times V(T),\rr \in V(T),\allowbreak \wt : V(T) \to \mathbb{N})$ such that
  $T$ is a subdivided star with root $\rr$
  and $\ell \geq 3$ leaves.  Suppose all the terminal pair paths
  in $\mathcal{P}$ pass through $\rr$. Then
  one can find the weight of a $\mathcal{P}$-multicut $S \subseteq V(T)$
  such that $\wt(S)$ is minimum,
  in time \timedpcarcless.
\end{proposition}
\begin{proof}
  We show how to reduce the given instance to an equivalent arcless instance of \wdpc.
  Let $u_1, \ldots, u_{\ell}$ be the leaves of $T$
  and let $P_i$ be the path from $\rr$ to $u_i$.
  For each $i \in [\ell]$, guess whether
  the solution of $\mathcal{I}$ intersects $P_i \setminus \{\rr\}$.
  This takes $2^{\ell}$ branches.
  In the branch corresponding to the guess
  where $S \cap (V(P_i) \setminus \{\rr\}) = \emptyset$,
  contract all the edges of $P_i$ onto the vertex $\rr$.
  Observe that the resulting graph in each of the branches is still a subdivided star
  and has the property that the solution intersects every root-to-leaf path.
  Further note that no minimal solution contains more than one vertex from the root-to-leaf path,
  otherwise the deletion of a vertex that is furthest from the root
  would result in a smaller solution.
  Thus, there exists a minimal solution that intersects every root-to-leaf path exactly once.
  Let $T'$ be constructed from $T$
  by (1) adding a new vertex $\bar{\rr}$
  and making it an out-neighbour of all the leaves of $T$,
  and (2) orienting every edge $xy \in E(T)$ from $x$ to $y$ if $y$ is closer to $\bar{\rr}$ than $x$ (in the undirected setting).
  Then clearly $T'$ is a graph
  consisting of internally vertex-disjoint $(\rr,\bar{\rr})$-paths
  such that the solution intersects each vertex-disjoint path exactly once.
  We define a weight function $\wt' : E(T') \to \mathbb{N}$ as follows.
  For every $i \in [\ell]$,
  let  $P_i=(v_{i,1}, \ldots , v_{i,p_i})$ where $\rr= v_{i,1}$ and $u_i= v_{i,p_i}$.
  Then $\wt'((v_{i,j}, v_{i,j+1})) = \wt(v_{i,j+1})$ for each $j \leq p_i-1$, and $\wt'((v_{i,p_i}, \bar{\rr})) = L$
  where $L = \sum_{v \in V(T)} \wt(v) +1$.
  Then observe that $V'=\{v_{1,i_1}, \ldots, v_{\ell,i_{\ell}}\}$
  is a $\mathcal{P}$-multicut in $T$
  with $\wt(V')=\ww$ if and only if $E'=\{(v_{1,i_1-1},v_{1,i_1}), \ldots, (v_{\ell,i_{\ell}-1},v_{\ell,i_{\ell}})\}$
  is a $\mathcal{P}$-dpc in $T'$ with $\wt'(E') =\ww$.
\end{proof}


\begin{figure}[t]

\begin{minipage}{0.43\textwidth}
\centering

\begin{tikzpicture}
  [
  ]
  \node[stan,label=right:$\rr$] at (0,0) (root) {};
  \node[stan,changed,label=left:$t_2$, below left = of root] (t2) {};
  \node[stan,changed,label=left:$t_1$, below left = of t2] (t1) {};
  \node[stan,changed,label=left:$s_1$, below left = of t1] (s1) {};
  \node[stan,del,label=left:$s_2$, below left = of s1] (s2) {};
  \node[stan,del,label=left:$t_4$, below left = of s2] (t4) {};
  \node[stan,del,label=left:$s_3$, below left = of t4] (s3) {};
  \node[stan,del,label=left:$s_4$, below left = of s3] (s4) {};

  \draw[-,thick] (root) -- (t2);
  \draw[-,thick] (t2) -- (t1);
  \draw[-,thick] (t1) -- (s1);
  \draw[-,thick, del] (s1) -- (s2);
  \draw[-,thick, del] (s2) -- (t4);
  \draw[-,thick, del] (t4) -- (s3);
  \draw[-,thick, del] (s3) -- (s4);

  \node[stan, below right = of root] (leg3) {};
  \node[stan,label=right:$t_3$, below right = of leg3] (t3) {};
  \node[stan,label=right:$t_5$, below right = of t3] (t5) {};
  \node[below right = of t5] (leg3end) {};

  \draw[-,thick] (root) -- (leg3);
  \draw[-,thick] (leg3) -- (t3);
  \draw[-,thick] (t3) -- (t5);
  \draw[-,thick, dashed] (t5) -- (leg3end);

  \node[stan] (leg4) at (leg3end |- leg3) {};
  \node[stan,label=right:$t_6$, below right = of leg4] (t6) {};
  \node[below right = of t6] (leg4end) {};

  \draw[-,thick] (root) -- (leg4);
  \draw[-,thick] (leg4) -- (t6);
  \draw[-,thick, dashed] (t6) -- (leg4end);

  \node[stan,label=left:$s_6$] at (t4 |- t2) (s6) {};
  \node[stan,label=left:$s_5$, below left = of s6] (s5) {};
  \node[below left = of s5] (leg1end) {};

  \draw[-,thick] (root) -- (s6);
  \draw[-,thick] (s6) -- (s5);
  \draw[-,thick, dashed] (s5) -- (leg1end);


  \node[] (z) at (s1.east -| t2) {$z$};
  \draw[->] (z) -- (s1);

\end{tikzpicture}

\end{minipage}%
\hfill
\begin{minipage}{.55\textwidth}

\caption{
Updating the weights of the vertices in Lemma~\ref{lem:wt-mc-sub-star}.
$\calP = \{(s_i, t_i) \mid i \in [6]\}$,
$z$ is closest to $r$
such that the $(s_1,t_1)$-path is contained in the $(z,\rr)$-path.
\\
The dotted parts are deleted
and the weights of the filled vertices
also include the weight of the minimum weight solution below them.
}

\label{fig:path-cleaning}
\end{minipage}
\end{figure}
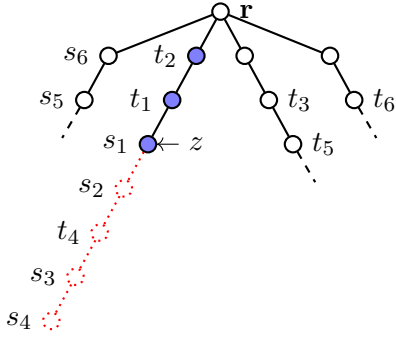

\begin{lemma}\label{lem:wt-mc-sub-star}
  \wmc\
  can be solved in \timedpcarcless-time
  on a subdivided star with $\ell$ leaves.
\end{lemma}
\begin{proof}
  Let the input instance be $\mathcal{I}=(T, \mathcal{P},\wt,\ww)$.
  If $\ell=2$, then $T$ is a path. In this case, report \yes\
  if and only if $\Apath(T, \mathcal{P},\wt) \leq \ww$.
  Otherwise, let $\rr \in V(T)$ be the root of $T$, that is $\rr$ is the unique vertex
  of degree at least $3$ in $T$.
  In the first step, the algorithm guesses
  whether $\rr$ is in the solution or not.
  If $\rr$ belongs to the solution,
  then delete $\rr$ from $T$ and
  solve the resulting instance using $\Apath$.
  Formally, the algorithm returns \yes\ if and only if
  $\wt(\rr) + \Apath(T,\calP,\wt) \leq \ww$.
  Henceforth, we assume that the solution does not contain $\rr$, or equivalently, we set
  $\wt(\rr) = \ww +1$.
  The remaining algorithm has two phases. In the first phase,
  it eliminates all the paths in $\mathcal{P}$ that do not pass through $\rr$.
  In the second phase, it uses the algorithm of Proposition~\ref{prop:verwdpc-arcless} to solve the problem.

  Suppose that there exists a path in $\mathcal{P}$ that does not pass through $\rr$.
  Let $z \in V(T)$ be a vertex that is closest to $\rr$
  such that there exists a path $P_{s,t}$ in $\mathcal{P}$
  where $P_{s,t} \subseteq P^{\dag}_{\rr,z}$.
  We create a new instance $\mathcal{I}' = (T', \mathcal{P}',\wt',\ww)$ (in polynomial time)
  such that $\mathcal{I}'$ is equivalent to $\mathcal{I}$.  Here, $T'= T \setminus T^{\dag}_z$ and,  $\mathcal{P}'=\mathcal{P} \setminus (V(T^{\dag}_{\rr,z}) \times V(T^{\dag}_{\rr,z})) \cup \{(\rr,z)\}$.
  Observe that the new terminal pair path $P_{\rr,z}$
  in $\mathcal{P}'$ intersects $\rr$ and
  thus, $\mathcal{P}'$ contains strictly fewer paths that do not pass through $\rr$
  (compared to $\mathcal{P}$).
  Since $T$ is a subdivided star,
  for each $v \in V(T) \setminus \{\rr\}$, $T^{\dag}_v$ is a path.
  The new weight function $\wt'$ is defined as follows (see Figure~\ref{fig:path-cleaning}).
  \[
    \wt'(v) =
    \begin{cases}
    \wt(v) + \Apath(T^{\dag}_v, \mathcal{P}|_{T^{\dag}_v}, \wt|_{V(T^{\dag}_v)})
    & \text{if } v \in V(P^{\dag}_{\rr,z})\\
    \wt(v) & \text{otherwise.}\\
    \end{cases}
  \]

  $(\Rightarrow)$ Let $S$ be a $\mathcal{P}$-multicut of $T$ such that
  $\wt(S) \leq \ww$.
  Since $P^{\dag}_{\rr,z}$ contains a path of $\mathcal{P}$,
  $S \cap V(P^{\dag}_{\rr,z}) \neq \emptyset$.
  Let $y \in S \cap V(P^{\dag}_{\rr,z})$ be the vertex that is closest to $\rr$.
  Construct $S' = S \setminus V(T^{\dag}_y)$.
  We claim that $S'$ is a solution for $\mathcal{I}'$.
  Observe that $S' \cap V(T^{\dag}_{\rr,z}) = \{y\}$.
  Observe that $S \cap V(T^{\dag}_y)$ is a $\mathcal{P}|_{T^{\dag}_y}$-multicut in $T^{\dag}_y$.
  Thus, $\wt(S \cap V(T^{\dag}_y)) \geq \Apath(T^{\dag}_y,\mathcal{P}|_{T^{\dag}_y},\wt|_{V(T_y^\dag)})$.
  From the construction of $S'$ and the weight function
  $\wt'$, $\wt'(S') = \wt(S) - \wt(S \cap V(T^{\dag}_y)) - \wt(y) + \wt'(y) \leq \wt(S) \leq \ww$.
  We now show that $S'$ is a $\mathcal{P}'$-multicut.
  Since $y \in S' \cap V(P^{\dag}_{\rr,z})$, $T-S'$ has no $(\rr,z)$-path.
  Consider any path of $\mathcal{P}'$ that intersects a vertex of $T^{\dag}_y$.
  Since the paths of $\mathcal{P}'$ are not contained in $T^{\dag}_{\rr,z}$,
  such a path also pass through $\rr$ and hence $y$. Since $y \in S'$, $S'$ is a $\mathcal{P}'$-multicut.

  $(\Leftarrow)$ Let $S'$ be a minimal $\mathcal{P}'$-multicut in $T$ such that $\wt'(S') \leq \ww$. Then $T-S'$ has no $(\rr,z)$-path.
  Since $\wt'(\rr)=\ww+1$, $S' \cap V(P^{\dag}_{\rr,z}) \neq \emptyset$. Since $S'$ is a minimal solution, $|S' \cap V(P^{\dag}_{\rr,z})|=1$
  for otherwise, deleting the vertex of $S$ on the $(\rr,z)$-path
  that is furthest from $\rr$
  would result in a smaller solution.
  Let $S' \cap V(P^{\dag}_{\rr,z}) =\{y\}$.
  Let $S^*$ be a minimum weight $\mathcal{P}|_{T^{\dag}_y}$-multicut.
  Then $\wt(S^*) = \Apath(T^{\dag}_y,\mathcal{P}|_{T^{\dag}_y},\wt|_{V(T_y^\dag)})$.
  Construct $S = S' \cup S^*$.
  We will now show that $S$ is a solution of $\mathcal{I}$.
  From the construction of $S$ and
  $\wt'$, $\wt(S) = \wt(S')+ \wt(S^*) = \wt'(S') - \wt'(y) + \wt(y) + \wt(S^*) \leq \wt'(S') \leq \ww$.
  Since $S' \subseteq S$, $S$ is a $\mathcal{P}'$-multicut.
  Consider a path of $\mathcal{P}$ that is contained in $T^{\dag}_{\rr,z}$.
  If such a path passes through $y$ or is contained in $T_y$,
  then it is intersected by $S^* \cup \{y\}$ (and hence $S$).
  Otherwise such a path is contained in $P^{\dag}_{\rr,y} \setminus \{y\}$.
  But this contradicts the choice of $z$.
  Therefore, $S$ is indeed a $\mathcal{P}$-multicut.

  We conclude that whenever there exists a path in $\mathcal{P}$
  that does not pass through $\rr$,
  we can apply the above procedure in polynomial time.
  Since every application of the above procedure
  decreases the number of paths of $\mathcal{P}$
  that do not pass through $\rr$ by at least one,
  the above procedure can be exhaustively applied in polynomial time.
  This ends the first phase of the algorithm. 
  At the end of the first phase, all the paths of $\mathcal{P}$ pass through $\rr$.
  Therefore, in this case, we solve the instance $(T,\mathcal{P},\rr,\wt)$
  using the algorithm of Proposition~\ref{prop:verwdpc-arcless}.
  Since the first phase of the algorithm takes polynomial time
  and the second phase takes \timedpcarcless-time,
  the algorithm runs in time \timedpcarcless.
\end{proof}

Observe that we can use Lemma~\ref{lem:wt-mc-sub-star}
to find the minimum weight $\mathcal{P}$-multicut in a subdivided star
by doing a simple binary search starting with $\ww=0, 1, 2, 4, 8, \ldots$ and so on.
This would incur an extra $\OO(\log \ww)$ factor in the running time.
Thus, even if $\ww$ is given as a unary input,
the resulting algorithm is still polynomial in the input size.
Therefore, the following corollary follows from
Lemmas~\ref{lem:wt-mc-path} and~\ref{lem:wt-mc-sub-star}.

\begin{corollary}\label{cor:min-wt-mc-sub-star}
  Let $T$ be a subdivided star with $\ell$ leaves.
  Let $\mathcal{P} \subseteq V(T) \times V(T)$ and
  $\wt : V(T) \to \mathbb{N}$.
  There is an algorithm $\Astar$
  that finds the weight of a $\mathcal{P}$-multicut $S \subseteq V(T)$
  such that $\wt(S)$ is minimum,
  in \timedpcarcless-time.
\end{corollary}

We are now equipped to design the branching algorithm for Theorem~\ref{thm:wt-mc-trees-l}.
Let $\mathcal{I}=(T,\mathcal{P}, \wt,\ww)$ be an instance of \wmctrees.
Root $T$ at an arbitrary vertex $\rr$.
With each instance $\mathcal{I}$,
we associate the measure $\mu(\mathcal{I})= |V_{\geq 3}(T)| + |V_{=1}(T)|$.
Since $|V_{=1}(T)| \leq \ell$ and $|V_{\geq 3}(T)| \le |V_{=1}(T)|-1$,
$\mu(\mathcal{I}) \leq 2 \ell$.
We now design a branching algorithm
such that the measure $\mu$ drops in each branch.
The following cases appear as base cases: (1) If $|V_{\geq 3}(T)| \leq 1$, then return \yes\ if and only if $\Astar(T,\mathcal{P},\wt) \leq \ww$, and (2) If $\ww <0$ or,  $\ww \leq 0$ and $\mathcal{P}\neq \emptyset$, then return \no.


If $|V_{\geq 3}(T)| \geq 2$,
let $x, y \in V_{\geq 3}(T)$ such that $x$ is a furthest in $T$ and, 
$y$ is its unique closest ancestor.
We branch into the following two cases. \\

\noindent
\textbf{Case 1.}
\emph{There exists a solution of $\mathcal{I}$ that does not intersect $V(P_{y,x})$.}
In this branch, we return the instance $\mathcal{I}_1=(T_1,\mathcal{P}_1,\wt_1,\ww)$
where $T_1= T/ E(P_{y,x})$ and $\mathcal{P}_1 =\mathcal{P}/E(P_{y,x})$.
Let the vertex onto which the edges of $P_{y,x}$ are contracted be $y^{\circ}$.
The new weight function $\wt_1$ is defined as follows: $\wt_1(v) =\wt(v)$
for each $v \in V(T_1) \setminus \{y^{\circ}\}$,
and $\wt_1(y^{\circ})=\ww+1$.
Observe that $\mathcal{I}_1$ can be constructed in polynomial time.
Furthermore, since $x,y \in V_{\geq 3}(T)$
and the edges of $P_{y,x}$ are contracted in
$\mathcal{I}_1$, $|V_{\geq 3}(T_1)| = |V_{\geq 3}(T)| -1$
and thus, $\mu(\mathcal{I}_1)=\mu(\mathcal{I})-1$. \\

\noindent
\textbf{Case 2.}
\emph{There exists a solution of $\mathcal{I}$ that intersects $V(P_{y,x})$.}
In this case,
let $z \in V(P_{y,x})$ be the closest vertex to $y$ such that $P_{y,z}$ contains a path of $\mathcal{P}$.
If no such vertex exists then set $z =x$.
Return the instance $\mathcal{I}_2=(T_2, \mathcal{P}_2,\wt_2,\ww)$
where $T_2 = T \setminus T_z^\dag$
and $\mathcal{P}_2=(\mathcal{P} \setminus (V(T_{y,x}) \times V(T_{y,x}))) \cup \{(y,z)\}$.
Observe that, by construction, any solution of $\mathcal{I}$ intersects $V(P_{y,z})$.
The new weight function $\wt_2$ is defined as follows (see Figure~\ref{fig:wt-update-fpt-ell}).
%
%
%
\[
  \wt_2(v) =
  \begin{cases}
    \wt(v) + \Astar(T^{\dag}_{v,x}, \mathcal{P}|_{T^{\dag}_{v,x}},\wt|_{V(T^{\dag}_{v,x})})
      & \text{ if } v \in V(P_{y,z})\\
    \wt(v) & \text{ otherwise.}
  \end{cases}
\]

Observe that, for each $v \in V(P_{y,x}) \setminus \{x\}$,
$T^{\dag}_{v,x}$ has exactly one branching vertex, namely $x$,
since $x$ is a furthest branching vertex in $T$ from $\rr$,
$y$ is the branching vertex that is the closest ancestor of $x$
and $v \in V(P_{y,x})$.  Also, $T^{\dag}_{x,x} = T^{\dag}_x$ is a disjoint union of paths.
Since $x \in V_{\geq 3}(T)$,
from the construction of $T_2$, $|V_{=1}(T_3)|<|V_{=1}(T)|$
and so, $\mu(\mathcal{I}_2) < \mu(\mathcal{I})$.

\begin{figure}[t]
\centering
\begin{subfigure}{0.48\textwidth}
  \centering
  \begin{tikzpicture}

  \node[] at (0,0) (root) {};
  \node[stan, below left = of root] (root1) {};
  \node[stan,contract,label=left:$\bm y$, below left = of root1] (y) {};
  \node[stan,contract, below left = of y] (t2) {};
  \node[stan,contract, below left = of t2] (t1) {};
  \node[stan,contract,label=left:$s_1$, below  = of t1] (s1) {};
  \node[stan,contract,label=left:$s_2$, below  = of s1] (s2) {};
  \node[stan,contract,label=left:$\bm x$, below  = of s2] (x) {};
  \node[stan,label=left:$t_3$, below left = of x] (t3) {};
  \node[stan,label=left:$s_3$, below  = of t3] (s3) {};
  \node[stan,label=left:$s_4$, below  = of s3] (s4) {};
  \node[stan, below right = of x] (x2) {};
  \node[stan,label=right:$t_4$, below  = of x2] (t4) {};
  \node[stan, below right = of y] (y1) {};
  \node[stan, below right = of y1] (y2) {};
  \node[stan, below  = of y2] (y3) {};
  \node[stan, below left = of y3] (y4) {};
  \node[stan, below right = of y3] (y5) {};

  \draw[-,thick, dotted] (root) -- (root1);
  \draw[-,thick] (root1) -- (y);
  \draw[-,thick,contract] (y) -- (t2);
  \draw[-,thick,contract] (t2) -- (t1);
  \draw[-,thick,contract] (t1) -- (s1);
  \draw[-,thick,contract] (s1) -- (s2);
  \draw[-,thick,contract] (s2) -- (x);
  \draw[-,thick] (x) -- (t3);
  \draw[-,thick] (t3) -- (s3);
  \draw[-,thick] (s3) -- (s4);
  \draw[-,thick] (x) -- (x2);
  \draw[-,thick] (x2) -- (t4);
  \draw[-,thick] (y) -- (y1);
  \draw[-,thick] (y1) -- (y2);
  \draw[-,thick] (y2) -- (y3);
  \draw[-,thick] (y3) -- (y4);
  \draw[-,thick] (y3) -- (y5);


  \end{tikzpicture}
  \caption{
  Case~1.
  The marked orange edges are contracted onto the undeletable vertex $y^\circ$.
  \\
  }
\end{subfigure}
\hfill
\begin{subfigure}{0.48\textwidth}
  \centering
  \begin{tikzpicture}

  \node[] at (0,0) (root) {};
  \node[stan, below left = of root] (root1) {};
  \node[stan,changed,label=left:$\bm y$, below left = of root1] (y) {};
  \node[stan,changed,label=left:$t_1$, below left = of y] (t2) {};
  \node[stan,changed,label=left:$t_2$, below left = of t2] (t1) {};
  \node[stan,changed,label=left:$s_1$, below  = of t1] (s1) {};
  \node[stan,del,label=left:$s_2$, below  = of s1] (s2) {};
  \node[stan,del,label=left:$\bm x$, below  = of s2] (x) {};
  \node[stan,del,label=left:$t_3$, below left = of x] (t3) {};
  \node[stan,del,label=left:$s_3$, below  = of t3] (s3) {};
  \node[stan,del,label=left:$s_4$, below  = of s3] (s4) {};
  \node[stan,del, below right = of x] (x2) {};
  \node[stan,del,label=right:$t_4$, below  = of x2] (t4) {};
  \node[stan, below right = of y] (y1) {};
  \node[stan, below right = of y1] (y2) {};
  \node[stan, below  = of y2] (y3) {};
  \node[stan, below left = of y3] (y4) {};
  \node[stan, below right = of y3] (y5) {};

  \draw[-,thick, dotted] (root) -- (root1);
  \draw[-,thick] (root1) -- (y);
  \draw[-,thick] (y) -- (t2);
  \draw[-,thick] (t2) -- (t1);
  \draw[-,thick] (t1) -- (s1);
  \draw[-,thick,del] (s1) -- (s2);
  \draw[-,thick,del] (s2) -- (x);
  \draw[-,thick,del] (x) -- (t3);
  \draw[-,thick,del] (t3) -- (s3);
  \draw[-,thick,del] (s3) -- (s4);
  \draw[-,thick,del] (x) -- (x2);
  \draw[-,thick,del] (x2) -- (t4);
  \draw[-,thick] (y) -- (y1);
  \draw[-,thick] (y1) -- (y2);
  \draw[-,thick] (y2) -- (y3);
  \draw[-,thick] (y3) -- (y4);
  \draw[-,thick] (y3) -- (y5);


  \node[] (z) at (s1.east -| y) {$z$};
  \draw[->] (z) -- (s1);

  \end{tikzpicture}
  \caption{
  Case~3.
  $z$ is closest to $y$ such that the $(y,z)$-path contains $(s_1,t_1)$.
  The weight of the filled vertices
  includes the weight of the minimum weight solution below them.
  $T^\dag_z$ is deleted (dotted part).
  }
\end{subfigure}

\caption{
The branches of the algorithm of Theorem~\ref{thm:wt-mc-trees-l}
with terminal pairs $(s_i,t_i)$. 
$x$ is a furthest branching vertex
and $y$ its unique closest branching ancestor.
}
\label{fig:wt-update-fpt-ell}
\end{figure}
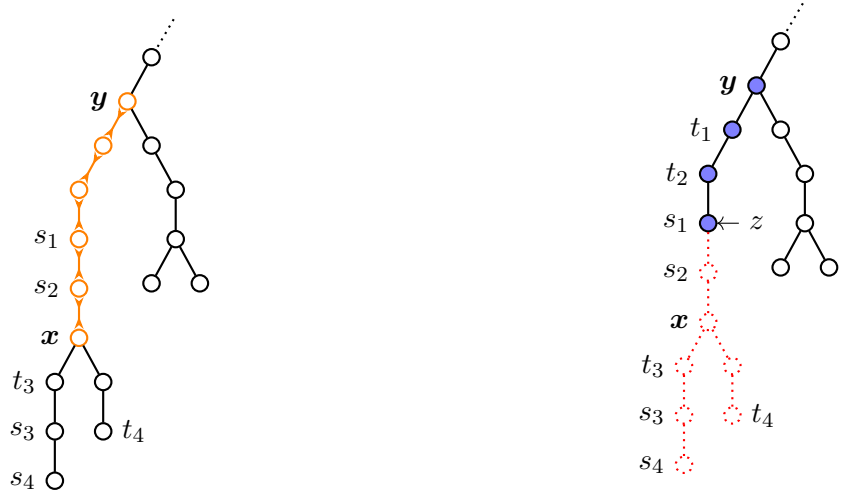

\begin{lemma}\label{lem:correctness-wt-mc-trees-l}
  $\mathcal{I}$ is a \yes-instance if and only if $\mathcal{I}_1$ or $\mathcal{I}_2$ is a \yes-instance.
\end{lemma}
\begin{proof}
  ($\Rightarrow$) Let $\mathcal{I}$ be a \yes-instance and let $S$ be a solution of $\mathcal{I}$.
  Suppose first that $S \cap V(P_{y,x}) = \emptyset$ and
  recall that $y^\circ$ is the vertex onto which the path $P_{y,x}$ is contracted in $\mathcal{I}_1$.
  Consider a path $P_{s,t}$ in $\mathcal{P}_1$.
  Then $(s,t) \neq (y^\circ,y^\circ)$ for otherwise,
  $S$ would not intersect the path in $\mathcal{P}$ corresponding to $(s,t)$.
  If $y^\circ \notin \{s,t\}$ then $(s,t) \in \mathcal{P}$ and so, $S$ intersects the path $P_{s,t}$.
  Otherwise, assume, without loss of generality, that $s =y^\circ$ and
  let $(z,t) \in \mathcal{P}$ where $z \in V(P_{y,x})$, be the terminal pair in $\mathcal{P}$ corresponding to $(s,t)$.
  Then since $P_{z,t}$ is intersected by $S \setminus V(P_{y,x})$, we conclude that $P_{s,t}$ is also intersected by  $S \setminus \{y^\circ\}$.

  Now suppose that $S \cap V(P_{y,x}) \neq \emptyset$.
  From the choice of $z$,
  $S \cap V(P_{y,z}) \neq \emptyset$.
  Let $v$ be the closest vertex of $P_{y,z}$ to $y$
  that belongs to $S$.
  Construct $S' = S \setminus V(T^{\dag}_{v,x})$.
  We claim that $S'$ is a $\mathcal{P}_2$-multicut in $T_2$ and
  $\wt_2(S') \leq \ww$.
  Since $v \in V(P_{y,z}) \cap S'$,
  $T_2 - S'$ does not contain the $(y,z)$-path.
  Consider now a path $P_{s,t}$ in $\mathcal{P}_2 \setminus \{(y,z)\}$.
  Then by definition of $\mathcal{P}_2$, $|\{s,t\} \cap V(T_{y,x})| \leq 1$.
  If $\{s,t\} \cap V(T_{y,x}) = \emptyset$ then $P_{s,t}$ is intersected by $S \setminus V(T_{y,x}) \subseteq S'$ since $S$ is a $\mathcal{P}$-multicut.
  Thus, suppose that $\{s,t\} \cap V(T_{y,x}) \neq \emptyset$, say $s \in V(T_{y,x})$ without loss of generality
  (note that then, $t \in V(T) \setminus V(T_{y,x})$).
  If $P_{s,t}$ contains $v$, then $P_{s,t}$ is intersected by $S'$ since $v \in S'$.
  Otherwise, $s \in V(P_{y,v}) \setminus \{v\}$ in which case, by the choice of $v$ and because $S$ is a $\mathcal{P}$-multicut, $P_{s,t}$ is intersected by $S \setminus V(T_{y,x}) \subseteq S'$.
  Thus, we conclude that $S'$ is a $\mathcal{P}_2$-multicut in $T_2$.
  From the construction of $S'$,
  observe that $S' \cap V(P_{y,z}) = \{v\}$.
  Thus, $\wt_2(S')=\wt(S) - \wt(S \cap V(T^{\dag}_{v,x})) - \wt(v) + \wt_2(v)$.
  Since $S \cap V(T^{\dag}_{v,x})$ is a $\mathcal{P}|_{T^{\dag}_{v,x}}$-multicut,
  $\wt(S \cap V(T^{\dag}_{v,x})) \geq \Astar(T^{\dag}_{v,x},\mathcal{P}_{T^{\dag}_{v,x}}, \wt|_{V(T^{\dag}_{v,x})})$. Therefore, $\wt_2(S') \leq \wt(S) \leq \ww$.

  ($\Leftarrow$)
  If $\mathcal{I}_1$ is a \yes-instance, then since $\wt_1(y^{\circ}) = \ww+1$, no solution of $\mathcal{I}_1$ contains $y^{\circ}$.
  Thus, every solution of $\mathcal{I}_1$ is also a solution for $\mathcal{I}$.
  If $\mathcal{I}_2$ is a \yes-instance and
  let $S'$ be a minimal solution of $\mathcal{I}_2$.
  Since $S'$ is a $\mathcal{P}_2$-multicut and
  $(y,z) \in \mathcal{P}_2$, $S' \cap V(P_{y,z}) \neq \emptyset$.
  In fact, since $S'$ is a minimal, $|S' \cap V(P_{y,z})| = 1$
  for otherwise, deleting the vertex of $S'$ on the $(y,z)$-path that is furthest from $y$ would result in a smaller solution
  (recall that $\mathcal{P}_2$ contains no terminal pair in $V(T_{y,x}) \times V(T_{y,x})$).
  Let $v \in S \cap V(P_{y,z})$.  If $v=x$, then $S'$ is also a solution for $\mathcal{I}$. Otherwise, let $v \neq x$.
  Let $S^*$ be the solution given by
  $\Astar(T^{\dag}_{v,x}, \mathcal{P}|_{T^{\dag}_{v,x}}, \wt|_{V(T^{\dag}_{v,x})})$ and
  let $S = S' \cup S^*$.
  We claim that $S$ is a $\mathcal{P}$-multicut in $T$ such that
  $\wt(S) \leq \ww$.
  Since $S' \subseteq S$,
  $S$ is a $\mathcal{P}_2$-multicut.
  Consider a path $P_{s,t}$ of $\mathcal{P}$ that is contained in $T_{y,x}$.
  Then either $P_{s,t}$ is contained in $T^{\dag}_{v,x}$ or $P_{s,t}$ contains $v$:
  indeed, if neither hold then $P_{s,t}$ is contained in $P_{y,v} \setminus \{v\}$, a contradiction to the choice of $z$.
  Now if $P_{s,t}$ is contained in $T^{\dag}_{v,x}$,
  then it is intersected by $S^*$;
  and if $P_{s,t}$ contains $v$, then it is intersected by $S$ since $v \in S$.
  Thus, $S$ is a $\mathcal{P}$-multicut.
  From the construction of $S$, $\wt(S)=\wt(S') + \wt(S^*)$.
  Also, $\wt(S')= \wt_2(S') - \wt_2(v)+\wt(v)$ because $S' \cap V(P_{y,z}) = \{v\}$.
  Since $\wt_2(v)=\wt(v)+\wt(S^*)$,
  we conclude that $\wt(S) = \wt_2(S') \leq \ww$.
\end{proof}

We now prove Theorem~\ref{thm:wt-mc-trees-l} formally.

\begin{proof}[Proof of Theorem~\ref{thm:wt-mc-trees-l}]
  Let $\mathcal{I}=(T,\mathcal{P},\wt,\ww)$ be an instance of \wmctrees.
  Lemma~\ref{lem:correctness-wt-mc-trees-l} shows that the algorithm described above correctly decides
  if $T$ has a $\mathcal{P}$-multicut $S$
  such that $\wt(S) \leq \ww$.
  Since the algorithm is a $2$-way branching algorithm,
  the measure of the algorithm,
  which drops by one in each branching step,
  is bounded by $2\ell$
  and the branching stops when the measure is at most $1$,
  the total number of nodes of the branching tree is at most $2^{2\ell +1}$.
  Now note that the worst running time at each branching node
  is during the construction of $\mathcal{I}_2$.
  However, since the construction of $\mathcal{I}_2$ requires making $\OO(n)$ calls
  to the algorithm of Corollary~\ref{cor:min-wt-mc-sub-star} ($\Astar$),
  $\mathcal{I}_2$ can be constructed in \timespecialdpc\ time.
  Thus, the algorithm runs in time \timedpcarcless.
\end{proof}

\section{Weighted Multicut on $(d,\ell)$-Light Instances}\label{sec:combines-dl}

\begin{figure}[t]
  \centering
  \includegraphics{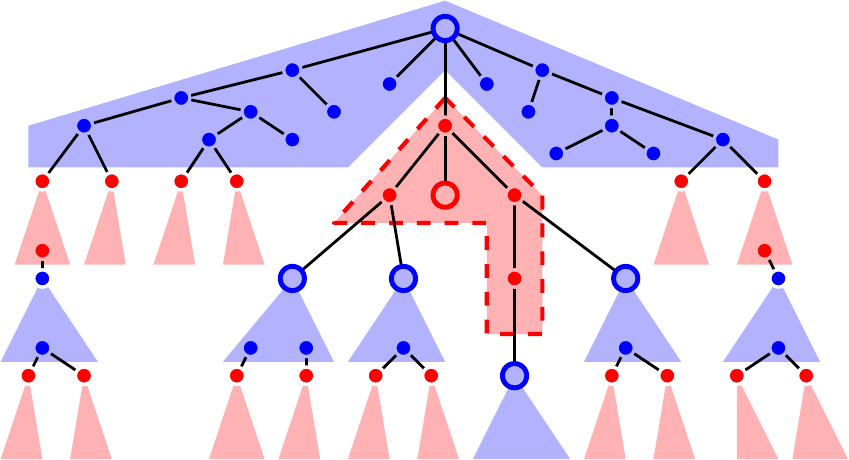}
  \caption{
  A $(d,6)$-light instance
  for some $d$.
  The $d$-light vertices vertices are shown in blue,
  and the $d$-heavy vertices in red.
  The closed neighborhood of the central component (marked with dashed boundary)
  containing $d$-heavy vertices
  has six leaves (marked with big circles),
  five of which are $d$-light vertices.
  }
  \label{fig:d-ell-light-example}
\end{figure}

In this section we prove Theorem~\ref{thm:wt-mc-trees-d-l}
by giving an \FPT-algorithm which solves \wmc on $(d,\ell)$-light instances
in time $3^d \cdot 2^{d\ell} \cdot 2^{\OO(\ell^2 \log \ell)} \cdot n^{\OO(1)}$.
For this, we first formally define the notion of $(d,\ell)$-light.

\begin{definition}
  \label{def:d-ell-light}
  Let $T$ be a tree
  and $\mathcal{P} \subseteq V(T) \times V(T)$ be a set of terminal pairs.
  A vertex $v \in V(T)$ is called a {\em $d$-light} vertex of $(T,\mathcal{P})$
  if at most $d$ terminal pair paths of $\mathcal{P}$ pass through $v$ in $T$.

  The set of $d$-light vertices of $(T,\mathcal{P})$
  is denoted by $\light(T,\mathcal{P},d)$.

  We say that $(T,\mathcal{P})$ is {\em $(d,\ell)$-light}
  if $T$ is a tree and for each connected component $C$
  of $T - \light(G,\mathcal{P},d)$,
  $N[C]$ has at most $\ell$ leaves.
\end{definition}
For ease of notation,
we say that a vertex $v \in V(T)\setminus \light(G, \calP, d)$
is \emph{$d$-heavy}.
See \cref{fig:d-ell-light-example} for an illustration of this definition.

\label{sec:appendix-extras}

In the definition of $(d,q)$-light instances, 
it is crucial to consider the number of leaves
in the tree induced by the \emph{closed neighborhood} of each component $C$ of $T-Y$ 
(i.e.\ $T[N[C]]$ must have at most $q$ leaves).
Assume, for a moment, that we just require that $T[C]$ has at most $q$ leaves,
then we do not expect the result as in Theorem~%
\ref{thm:wt-mc-trees-d-l}.

This is because with this new (and wrong) definition of $(d,q)$-light instances,
\wmctrees\ is \nph\ for $d=3$ and $q=2$.
Let $(G,\kk)$ be an instance of \textsc{Vertex Cover}
with $V(G) =\{v_1, \ldots, v_n\}$. 
Let $G'$ be a graph on $2n$ vertices
$x_1, \ldots, x_n, y_1, \ldots,y_n$
where $x_i$ is adjacent to $x_{i+1}$ for all $i \in [n-1]$
and each $y_i$ adjacent to $x_i$ for all $i\in [n]$.
Define the set of terminal pairs $\mathcal{P}$ as $(y_i,y_j) \in \mathcal{P}$
if and only if $(v_i,v_j) \in E(G)$.
The weight function $\wt : V(G') \to \{1,n+1\}$ is defined as
$\wt(x_i) = n+1$ for each $i \in [n]$ and $\wt(y_i)=1$ for each $i \in [n]$.
It is easy to observe that $\{v_{i_1}, \ldots, v_{i_\kk}\}$ is a vertex cover of $G$
if and only if
$\{y_{i_1}, \ldots, y_{i_\kk}\}$ is a $\mathcal{P}$-multicut in $G'$.
We know that \textsc{Vertex Cover} is \nph\ on graphs with maximum degree $3$
\cite{garey1974some}.
Hence, we can assume that
for each $y_i$ at most $3$ terminal pair paths of $\mathcal{P}$ pass through it. Thus $Y= \{y_1, \ldots,y_n\}$ is a set of $3$-light vertices and the removal of any superset of it leaves a collection of paths each of which has at most $2$ leaves.
Thus, with the {\em wrong} definition of $(d,q)$-light instances,
the resulting instance is a $(3,2)$-light instance.

\paragraph*{Notation.}
Let $\mathcal{I}=(T,\Pairs,\wt,\ww)$ be a \wmc instance.
Then $\calI$ is called a $(d,\ell)$-light instance
if $(T,\Pairs)$ is $(d,\ell)$-light.

Assume that $T$ is rooted.
For every vertex $v \in V(T)$,
we denote by $I[v] \subseteq \calP|_{T_v}$
the set of terminal pairs $(s,t) \in \calP|_{T_v}$
such that $v$ is contained in the $(s,t)$-path in $T$,
and by $O[v] \subseteq \mathcal{P}$
the set of terminal pairs $(s,t) \in \Pairs$
such that $\{s,t\} \cap V(T_v) \neq \emptyset$
and $\{s,t\} \cap V(T) \setminus V(T_v) \neq \emptyset$.
In other words, $I[v]$ denotes the set of terminal pairs going through $v$
and which are fully contained in the subtree rooted at $v$.
In constrast, the set $O[v]$ contains those terminal pairs
going through $v$ and leaving the subtree rooted at $v$.
Note that if $v$ is a $d$-light vertex then $|I[v]| + |O[v]| \leq d$ by definition.

\paragraph*{Intuition.}
The intuition of the algorithm is as follows.
For each $d$-light vertex $v \in V(T)$ and for all sets $O \subseteq O[v]$,
we compute the minimum weight of a partial solution $S_{O,v} \subseteq V(T_v)$
such that $S_{O,v}$ is a $\calP\vert_{T_v}$-multicut
and for every $(s,t) \in O$, $S_{O,v}$ intersects the $(s,t)$-path in $T$.
We store this minimum weight of a solution as $\Tab[v, O]$.
We use a dynamic program to compute the table entries for the $d$-light vertices
in a bottom-up transversal of $T$ (we assume that $T$ is rooted).
The crucial part of the algorithm is that
we do \emph{not compute these partial solutions} for \emph{every} $d$-heavy vertices.
Instead, one can think of partitioning the tree into (connected) components
corresponding to the status of being $d$-light or $d$-heavy.
For the components with the $d$-light vertices,
we compute the best solution by an exhaustive search.
This works because there are at most $d$ terminal pair paths going through a $d$-light vertex.
For the components consisting of $d$-heavy vertices,
we make use of the previous result in Theorem~\ref{thm:wt-mc-trees-l}
to compute a minimum solution.
Here, we exploit the fact that such a component has at most $\ell$ leaves.
We first design the main algorithm 
that computes the table entries $\Tab[v, \cdot]$
for every $v \in \light(T, \Pairs, d)$.
This algorithm uses as a subroutine the second algorithm $\Aheavyv$
to compute partial solutions for the $d$-heavy children of $v$.


First observe that if
$(T,\mathcal{P})$ contains no $d$-light vertex then, by definition,
$T$ has at most $\ell$ leaves and thus,
we may use the algorithm of Theorem~\ref{thm:wt-mc-trees-l}
to solve \wmc\ on instance $\mathcal{I}$ in time \timewmctrees.
Assume henceforth that $\light(T,\Pairs,d) \neq \emptyset$
and let us root $T$ at some $d$-light vertex $\rr$.
We define the table $\Tab$ formally as follows:
for every $v \in \light(T, \Pairs, d)$ and
for every set $O \subseteq O[v]$,
\begin{gather*}
  \label{eq:table}
  \Tab[v,O] \overset{\mathrm{def}}{=}
  \min_{S\subseteq V(T_v)} \wt(S)
  \text{ s.t. }
  \\
  \begin{align*}
      S \text{ is a } \calP|_{T_v}\text{-multicut}
      \land
      \forall (s,t)\in O: S \text{ intersects the } (s,t)\text{-path in } T
  \end{align*}
\end{gather*}

Initially, every entry of the table $\Tab$ is set to $+\infty$.
To update each entry of $\Tab[v,\cdot]$,
we assume that $\Aheavyv$ is given.
The output is \yes if $\Tab[\rr, \emptyset] \leq \ww$
and \no otherwise.
Note that this entry is defined as we assume that $\rr$ is $d$-light.
For every $d$-light vertex, we proceed as follows.


\paragraph*{\textbf{Updating $\bm d$-Light Leaves.}}
Let $v \in \light(T,\Pairs,d)$ be a leaf of $T$.
Then for every $O \subseteq O[v]$, we set
\[
  \Tab[v,O] =
    \begin{cases}
      \wt(v) & \text{if } O \neq \emptyset, \\
      0 & \text{otherwise}.
    \end{cases}
\]

\paragraph*{\textbf{Updating Internal $\bm d$-Light Vertices.}}
Let $v \in \light(T,\mathcal{P},d)$ be an internal vertex of $T$
and let $u_1,\ldots,u_p \in V(T)$ be the children of $v$.
Let $O \subseteq O[v]$ be fixed.
As mentioned above, we assume that there is a subroutine $\Aheavyv$
which takes as an input a child $u \notin \light(T,\mathcal{P},d)$ of $v$
and a set $Q \subseteq O[u]$ of terminal pairs,
and outputs the minimum weight of a set $S \subseteq V(T_u)$ such that
$S$ is a $\calP|_{T_u}$-multicut and for every $(s,t) \in Q$,
$S$ intersects the $(s,t)$-path in $T$
(we show below how to obtain $\Aheavyv$).
 For ease of notation, we define a function $\mathcal{A}_v^*$:
 for every child $u$ of $v$ and every set $Q \subseteq O[u]$,
 \[
   \mathcal{A}_v^*(u,Q) =
     \begin{cases}
         \Tab[u,Q] & \text{if } u \in \light(T,\mathcal{P},d), \\
         \Aheavyv(u,Q) & \mbox{otherwise.}
     \end{cases}
 \]

\paragraph*{Intuition.}
We first describe the intuition of the algorithm.
Note that it is always possible to delete $v$.
In this case, the solution is the disjoint union of optimal solutions for the children,
where we do not have to cut any of the outgoing terminal pairs.
Moreover, we have to delete $v$ if $(v,v)$ is a terminal pair, or there are terminal pairs in $O$
which start at $v$ and leave $T_v$ (these pairs are later denoted by $O_0$).

In the case where $v$ is not deleted, we proceed as follows.
We denote by $I_{i,j}$ the set of terminal pairs in $I[v]$
which use $u_i$ and $u_j$ where $i < j$.
Likewise,
the set $I_i \subseteq I[v]$ denotes the terminal pairs which use only $u_i$ and end at $v$,
i.e.\ they do not go into any other subtree.
We guess which of the paths in $I_{i,j}$ are cut by the solution for the subtree rooted at $u_i$.
We denote this set by $Q_{i,j}$.
Note that the pairs in $I_{i,j} \setminus Q_{i,j}$ must then be cut by the solution for the subtree rooted at $u_j$.
Besides these pairs, we also have to cut the terminal pairs
which leave $T_v$ and start/end in a subtree of some child $u_i$. 
We denote this set by $O_i$.
Thus,
for each child $u_i$ we have to cut the terminal pairs in
$O_i \cup I_i \cup \bigcup_{j>i} Q_{i,j} \cup \bigcup_{j<i} I_{j,i}\setminus Q_{j,i}$.\\

We now give the formal algorithm.
For every $i,j \in [p]$ where $i < j$,
denote by $I_{i,j} = I[v] \cap O[u_i] \cap O[u_j]$
and for every $i \in [p]$,
denote by $I_i = I[v] \cap O[u_i] \setminus \bigcup_{j>i} I_{i,j}$.
Further denote by $I_0 = \{(v,v)\} \cap I[v]$.
Note that for every $(s,t) \in I[v]$, one of three cases may arise:
\begin{itemize}
  \item
  $(s,t) \notin \bigcup_{i \in [p]} O[u_i]$, that is, $(s,t) = (v,v)$, or
  \item
  there exists a unique index $i \in [p]$ such that $(s,t) \in O[u_i]$, or
  \item
  there exist exactly two indices $i\in [p]$ such that $(s,t) \in O[u_i]$.
\end{itemize}
Indeed, the $(s,t)$-path would otherwise leave $T_v$
thereby contradicting the fact that $(s,t)\in I[v]$.
Therefore, $\{I_0,I_1,\ldots,I_p,I_{1,2},\ldots,I_{p-1,p}\}$ is a partition $I[v]$.

Denote by $O_0 \subseteq O$ the set of terminal pairs $(s,t) \in O$
such that $(s,t) \notin \bigcup_{i \in [p]} O[u_i]$,
and for every $i \in [p]$,
denote by $O_i \subseteq O$ the set of terminal pairs $(s,t) \in O$
such that $(s,t) \in O[u_i]$.
Note that for every $(s,t) \in O$, one of two cases may arise:
\begin{itemize}
  \item
  $(s,t) \in O[v] \setminus \bigcup_{i \in [p]} O[u_i]$, or
  \item
  there exists a unique index $i \in [p]$ such that $(s,t) \in O[u_i]$.
\end{itemize}
Indeed, the $(s,t)$-path would otherwise not leave $T_v$
thereby contradicting the fact that $(s,t)\in O[v]$.
Therefore, $\{O_0,\ldots,O_p\}$ is a partition $O$.

If $I_0 \cup O_0 \neq \emptyset$,
then the only way to separate the terminal pairs in $I_0 \cup O_0$
is by removing $v$.
Thus, in this case, we update $\Tab[v, O]$ as follows.
\[
  \Tab[v, O] =
  \wt(v) + \sum_{i \in [p]} \mathcal{A}_v^*(u_i,\emptyset).
\]
Otherwise, let a \emph{distribution} be a $p(p-1)/2$-tuple 
$\pi = (Q_{1,2},Q_{1,3},\ldots,Q_{p-1,p})$
where for every $i \in [p-1]$ and $j > i$,
$Q_{i,j}$ is a subset of $I_{i,j}$.
Then we update $\Tab[v, O]$ according to the following procedure.

\begin{enumerate}
  \item
  Set $\Tab[v, O] = \wt(v) + \sum_{i \in [p]} \mathcal{A}_v^*(u_i,\emptyset)$.
  \item
  \label{step:algoLight:distributionsLoop}
  For every distribution $\pi=(Q_{1,2},Q_{1,3},\ldots,Q_{p-1,p})$ do:
  \begin{itemize}
    \item
    For every $i \in [p]$, define
    $
      \Pairs^\pi_i = O_i \cup I_i
        \cup \bigcup_{i < j} Q_{i,j}
        \cup \bigcup_{i > j} I_{j,i} \setminus Q_{j,i}.
    $
    \item
    \(
      \Tab[v, O] = \min\{\Tab[v, O],
        \sum_{i \in [p]} \mathcal{A}_v^*(
          u_i, \Pairs^\pi_i
          )\}
    \)
  \end{itemize}
\end{enumerate}

\begin{lemma}
  For every internal $d$-light vertex $v$,
  if $\mathcal{A}_v^*$ is correct then
  the table entries $\Tab[v, \cdot]$ are updated correctly.
  Furthermore, $\Tab[v, \cdot]$ can be updated in
  $3^d \cdot \OO(\mathsf{T}(\Aheavyv))$-time
  where $\mathsf{T}(\Aheavyv)$ is the running time
  of the subroutine $\Aheavyv$.
\end{lemma}

\begin{proof}
  Let $v \in \light(T,\mathcal{P},d)$ be an internal vertex of $T$
  and let $u_1, \ldots,u_p \in V(T)$ be the children of $v$.
  Assume that $\mathcal{A}^*_v$ is correct
  (in particular, for every child $u \in \light(T,\mathcal{P},d)$ of $v$,
  the table $\Tab[u,\cdot]$ is correctly filled).

  Let us first show that for any set $O \subseteq O[v]$,
  there exists a set $S \subseteq V(T_v)$ of weight $\Tab[v, O]$ such that
  $S$ is a $\calP|_{T_v}$-multicut and for every $(s,t) \in O$,
  $S$ intersects the $(s,t)$-path in $T$.

  Consider a set $O \subseteq O[v]$.
  Observe that the update step sets $\Tab[v, O]$ to a finite value.
  Suppose first that $I_0 \cup O_0 \neq \emptyset$.
  Since $\mathcal{A}^*_v$ is correct,
  it follows from the update step that there exists for every $i \in [p]$,
  a $\calP|_{T_{u_i}}$-multicut $S_i$ such that
  $\Tab[v, O] = \wt(v) + \sum_{i \in [p]} \wt(S_i)$.
  Then the set $\{v\} \cup \bigcup_{i \in [p]} S_i$ is the desired $S$.
  Second, suppose that $I_0 \cup O_0 = \emptyset$.
  Since $\mathcal{A}^*_v$ is correct,
  it follows from the update step that either
  \begin{enumerate}[(i)]
    \item
    \label{step:algoLight:proof:first}
    there exists for every $i \in [p]$,
    a $\calP|_{T_{u_i}}$-multicut $S_i$ such that
    $\Tab[v, O] = \wt(v) + \sum_{i \in [p]} \wt(S_i)$, or
    \item
    \label{step:algoLight:proof:second}
    there is a distribution $\pi=(Q_{1,2},Q_{1,3},\ldots,Q_{p-1,p})$ such that
    for every $i \in [p]$, there exists a set $S_i$ where
    \begin{itemize}
      \item
      $S_i$ is a $\calP|_{T_{u_i}}$-multicut and
      \item
      for every $(s,t) \in \Pairs^\pi_i$,
      $S_i$ intersects the $(s,t)$-path in $T$,
    \end{itemize}
    and $\Tab[v, O] = \sum_{i \in [p]} \wt(S_i)$.
  \end{enumerate}
  If \ref{step:algoLight:proof:first} holds
  then we conclude, as previously,
  that $\{v\} \cup \bigcup_{i \in [p]} S_i$ is the desired $S$.
  Thus, assume that~\ref{step:algoLight:proof:second} holds and let us show that
  $\bigcup_{i \in [p]} S_i$ is the desired $S$.
  Note that since for every $i \in [p]$,
  $S_i$ is a $\calP|_{T_{u_i}}$-multicut,
  it suffices to show that for every $(s,t) \in I[v] \cup O$,
  $\bigcup_{i \in [p]} S_i$ intersects the $(s,t)$-path in $T$.
  Consider, therefore, a terminal pair $(s,t) \in I[v] \cup O$.
  If $(s,t) \in O_i \cup I_i$ for some $i \in [p]$,
  then $S_i$ intersects the $(s,t)$-path in $T$ by definition.
  Thus, assume that $(s,t) \in I_{i,j}$ for some $i,j \in [p]$ where $i < j$.
  Then either $(s,t) \in Q_{i,j}$
  in which case $S_i$ intersects the $(s,t)$-path in $T$ by definition;
  or $(s,t) \in I_{i,j} \setminus Q_{i,j}$ in which case
  $S_j$ intersects the $(s,t)$-path in $T$ by definition.

  Second, let us show that for any set $O \subseteq O[v]$,
  if $S$ is a set of minimum weight such that $S$ is a $\calP|_{T_v}$-multicut
  and for every $(s,t) \in O$, $S$ intersects the $(s,t)$-path in $T$,
  then $\wt(S) \geq \Tab[v, O]$.

  Let $S$ be a set of minimum weight such that
  $S$ is a $\calP|_{T_v}$-multicut and for every $(s,t) \in O$,
  $S$ intersects the $(s,t)$-path in $T$.
  For every $i \in [p]$, denote by $S_i = S \cap V(T_{u_i})$.
  Suppose first that $v \in S$.
  Then for every $i \in [p]$,
  $S_i$ is a $\calP|_{T_{u_i}}$-multicut which implies that
  $\wt(S_i) \geq  \mathcal{A}_v^*(u_i,\emptyset)$
  as $\mathcal{A}^*_v$ is correct.
  But $\Tab[v,O] \leq \wt(v) + \sum_{i \in [p]} \mathcal{A}_v^*(u_i,\emptyset)$
  by the update step and thus, the claim holds true.
  Second, suppose that $v \notin S$
  (note that $I_0 \cup O_0 = \emptyset$ in this case).
  Then for every $i,j \in [p]$ where $i < j$,
  and every terminal pair $(s,t) \in I_{i,j}$,
  at least one of $S_i$ and $S_j$ must intersect the $(s,t)$-path in $T$:
  let us denote by $Q_{i,j} \subseteq I_{i,j}$
  the set of terminal pairs $(s,t) \in I_{i,j}$
  such that $S_i$ intersects the $(s,t)$-path in $T$
  (note that then, for every $(s,t) \in I_{i,j} \setminus Q_{i,j}$,
  $S_j$ intersects the $(s,t)$-path in $T$).
  Since the update procedure loops over every distribution,
  it considers at some point the distribution
  $(Q_{1,2},Q_{1,3},\ldots, Q_{p-1,p})$ and thus,
  \[
    \Tab[v, O] \leq \sum_{i \in [p]}
      \mathcal{A}_v^*(u_i,
        \Pairs^\pi_i
      )
  \]
  where we set
  $\Pairs^\pi_i = O_i \cup I_i
    \cup \bigcup_{i < j} Q_{i,j}
    \cup \bigcup_{i > j} I_{j,i} \setminus Q_{j,i}$
  as above.

  Now observe that for every $i \in [p]$
  and every terminal pair $(s,t) \in \calP|_{T_{u_i}} \cup O_i \cup I_i$,
  $S_i$ must intersect the $(s,t)$-path in $T$, as $v \notin S$,
  and so
  \[
    \wt(S_i) \geq
      \mathcal{A}_v^*(u_i,
        \Pairs^\pi_i
      )
  \]
  as $\mathcal{A}^*_v$ is correct.
  Therefore, $\Tab[v, O] \leq \wt(S)$ as claimed.

  By the above, we conclude that for every set $O \subseteq O[v]$,
  $\Tab[v, O]$ is filled correctly.
  Finally, observe that for every $O \subseteq O[v]$,
  \[
    |O| + \sum_{i \in [p-1]} \sum_{j>i} |I_{i,j}| \leq d
  \]
  since $v$ is $d$-light.
  Thus, for a fixed $O \subseteq O[v]$,
  there are
  \[
      \prod_{ i \in [p-1]} \prod_{j>i} 2^{|I_{i,j}|}
    = 2^{\sum_{i \in [p-1]} \sum_{j>i} |I_{i,j}|} \leq 2^{d - |O|}
  \]
  distributions to consider in Step~\ref{step:algoLight:distributionsLoop}.
  It follows that $\Tab[v, O]$ can updated in time at most
  $2^{d-|O|} \cdot \OO(\mathsf{T}(\Aheavyv))$
  and thus,
  it takes at most
  \[
    \sum_{i=0}^d {d \choose i} 2^{d-i}
      \cdot \OO(\mathsf{T}(\Aheavyv))
    = 3^d \cdot \OO(\mathsf{T}(\Aheavyv))
  \]
  time to update $\Tab[v, \cdot]$.
\end{proof}

\paragraph*{\textbf{Heavy Vertices with $\bm d$-Light Parents.}}
Let us now describe the subroutine $\Aheavyv$ 
for a $d$-light vertex $v$ with at least one $d$-heavy child.
Let $u$ be a fixed $d$-heavy child of $v$.
Denote by $C_u$ the connected component of $T - \light(T,\mathcal{P},d)$ containing $u$
and by $N_u = N(C_u) \cap \light(T,\mathcal{P},d) \setminus \{v\}$.
If $N_u \neq \emptyset$, then we denote by $N_u = \{u_1,\ldots,u_\lambda\}$ for some $\lambda\le \ell$, and for every $i \in [\lambda]$, we let
$p_i$ be the parent of $u_i$ in $T$ (note that, by definition, $p_i \in C_u$ for every $i \in [\lambda]$).

Given a set $Q \subseteq O[u]$,
the basic idea is to ``guess''
for each $u_i$,
a set $O_i$ of pairs in $O[u_i]$ which are already cut
by an optimal solution for the subtree $T_{u_i}$.
We are then only interested in separating terminal pairs which intersect $T_u$
and are not already separated by a solution in some $T_{u_i}$. 
By definition of the problem,
we also do not have to separate the pairs in $O[u]\setminus Q$.
By these observations, it suffices to only consider the subtree $T'_u$
obtained from $T_u$ after deleting all subtrees $T_{u_i}$.
Because of this deletion, there might be pairs which do not start or end in $T'_u$
but must be separated in $T'_u$.
We take care of those by constructing a projection $\tau$
which maps the start and end point of the terminal pairs
to the first vertex of the path which lies inside $T'_u$.

\begin{remark}\label{rem:minimum-sol-for-uncon-multicut}
  The algorithm of \cref{thm:wt-mc-trees-l} can,
  with an additional $\OO(\log \ww)$-factor in the runtime,
  find the minimum weight of a $\mathcal{P}$-multicut
  in a tree with $\ell$ leaves,
  in time \timewmctrees.
  We denote this algorithm by $\Aunmc$.
\end{remark}


A \emph{distribution} is a $\lambda$-tuple $(O_1,\ldots,O_\lambda)$
where for every $i \in  [\lambda]$, $O_i \subseteq O[u_i]$.
Given a set $Q \subseteq O[u]$,
the algorithm $\Aheavyv$ proceeds as follows.

\begin{enumerate}[1.]
\item
\label{step:algoHeavy:Nuempty}
If $N_u = \emptyset$ then return the weight of the set
\[
  \Aunmc(T_u,\calP|_{T_u} \cup \{(s,u)~|~(s,t) \in Q \text{ and } s \in V(T_u)\},\wt|_{V(T_u)}).
\]
\item
Initialize $\mathsf{OPT}=\infty$.
\item 
Set $T'_u = T[V(T_u) \setminus \bigcup_{i \in [\lambda]} V(T_{u_i})]$.
\item
Define the projection $\tau_u:V(T) \to V(T'_u)$
  where for all $v \in V(T)$
  \[
    \tau_u(v) = \begin{cases}
        p_i & \text{if } v \in V(T_{u_i}), \\
        u & \text{if } v \in V(T) \setminus V(T_u), \\
        v & \text{otherwise.}
      \end{cases}
  \]
\item
\label{step:algoHeavy:loop}
For every distribution $\pi = (O_1,\ldots,O_p)$ do:
\begin{enumerate}[label*=\arabic*.]
  \item
  Let
  \(
    \Pairs_{u,\pi} = (Q \cup \Pairs\vert_{T_u})
      \setminus \bigcup_{i\in[\lambda]} (\Pairs\vert_{T_{u_i}} \cup O_i)
  \).
\item
  Set
  \(
    \Pairs'_{u,\pi} = \{ (\tau_u(s), \tau_u(t))
      \mid (s,t) \in \Pairs_{u,\pi}
      \}
    .
  \)
  \item
  \label{step:algoHeavy:computingSolution}
  Compute
  \(
    M=\Aunmc(T'_u, \calP'_{u,\pi}, \wt|_{V(T'_u)})
  \)
  \item
  Set
  $\mathsf{OPT} = \min \{\mathsf{OPT}, \wt(M) + \sum_{i \in [\lambda]} \Tab[u_i, O_i]\}$.
\end{enumerate}
\item
Return $\mathsf{OPT}$. 
\end{enumerate}

\begin{lemma}
  For every $d$-light vertex $v$ with at least one $d$-heavy child,
  $\Aheavyv$ is correct and runs in time
  $2^{d\ell} \cdot 2^{\OO(\ell^2\log \ell)} \cdot n^{\OO(1)}$.
\end{lemma}

\begin{proof}
  Let $v$ be a $d$-light vertex in $T$ 
  with at least one $d$-heavy child and let $u$ be one such child of $v$.
  Consider a set $Q \subseteq O[u]$.
  If $N_u = \emptyset$ then, 
  denoting by $M = \Aunmc(T_u,\calP|_{T_u} \cup \{(s,u)~|~(s,t) \in Q \text{ and } s \in V(T_u)\},\wt|_{V(T_u)})$, 
  it is clear that $M$ is a subset of $V(T_u)$ of minimum weight such that
  $M$ is a $\calP|_{T_u}$-multicut and for every $(s,t) \in Q$,
  $M$ intersects the $(s,t)$-path in $T$. 
  Thus, $\Aheavyv$ indeed outputs the correct answer in this case.
  
  Suppose next that $N_u = \{u_1,\ldots,u_\lambda\}$ for some $\lambda \in [\ell]$,
  and assume that for every $i\in [\lambda]$,
  $\Tab[u_i, \cdot]$ is correctly filled.
  Let $S \subseteq V(T_u)$ be a set of minimum weight such that
  $S$ is a $\calP|_{T_u}$-multicut and for every $(s,t) \in Q$,
  $S$ intersects the $(s,t)$-path in $T$.
  For every $i \in [\lambda]$, let $S_i = S \cap V(T_{u_i})$
  and let $O_i \subseteq O[u_i]$ be the set of terminal pairs $(s,t)$
  such that $S_i$ intersects the $(s,t)$-path in $T$.

  Since the algorithm loops over every distribution,
  it considers at some point the distribution $\pi = (O_1,\ldots,O_\lambda)$:
  let $M^*= \Aunmc(T'_u, \Pairs'_{u,\pi}, \wt|_{V(T'_u)})$
  where $\Pairs'_{u,\pi}$
  and $T'_u$ are as defined in the algorithm.
  We aim to show that $\wt(S) = \wt(M^*) + \sum_{i \in [\lambda]} \Tab[u_i, O_i]$.

  To this end, let us show that $S^* = S \setminus \bigcup_{i \in [\lambda]} S_i$
  is a $\Pairs'_{u,\pi}$-multicut in $T'_u$.
  Consider a terminal pair
  $(s,t) \in \Pairs'_{u,\pi}$.
  If $(s,t) \in \Pairs_{u,\pi} \cap \Pairs'_{u,\pi}$,
  then $S^*$ intersects the $(s,t)$-path in $T'_u$
  since $S$ is a $\calP|_{T_u}$-multicut
  and $s,t \in V(T_u) \setminus \bigcup_{i \in [\lambda]} V(T_{u_i})$.

  Suppose next that $(s,t) \in \Pairs'_{u,\pi} \setminus \Pairs_{u,\pi}$.
  Then $(s,t)$ corresponds to a terminal pair
  $(a,b) \in \Pairs_{u,\pi} \setminus \Pairs'_{u,\pi}$
  such that $\tau_u(a)=s$ and $\tau_u(b)=t$. 
  Since by construction, $S$ intersects the $(a,b)$-path in $T$ and
  for every $i \in [\lambda]$, $(a,b) \notin O_i$,
  it follows that $S^*$ intersects the $(a,b)$-path in $T$
  and, a fortiori, the $(s,t)$-path in $T'_u$.

  Therefore, $S^*$ is a $\Pairs'_{u,\pi}$-multicut in $T'_u$ as claimed;
  in particular, $\wt(M^*) \leq \wt(S^*)$ by minimality of $M^*$.
  Now observe that for every $i \in [\lambda]$,
  $S_i$ is a $\calP|_{T_{u_i}}$-multicut such that for every $(s,t) \in O_i$,
  $S_i$ intersects the $(s,t)$-path in $T$,
  and thus, $\wt(S_i) \geq \Tab[u_i, O_i]$
  as $\Tab[u_i, \cdot]$ is correctly filled by assumption.
  It follows that $\wt(M^*) + \sum_{i \in [\lambda]} \Tab[u_i, O_i] \leq \wt(S)$:
  indeed, since $\wt(M^*) \leq \wt(S^*)$ and
  $\wt(S^*) = \wt(S) - \sum_{i \in [\lambda]} \wt(S_i)$,
  \[
    \wt(M^*) \leq \wt(S) - \sum_{i \in [\lambda]} \wt(S_i)
      \leq \wt(S) - \sum_{i \in [\lambda]} \Tab[u_i, O_i].
  \]

  To prove the converse inequality, for every $i \in [\lambda]$,
  let $S^*_i$ be a set of minimum weight such that
  $S^*_i$ is a $\calP|_{T_{u_i}}$-multicut
  and for every $(s,t) \in O_i$, $S^*_i$ intersects the $(s,t)$-path in $T$.
  We contend that $M = M^* \cup \bigcup_{i \in [\lambda]} S^*_i$
  is a $\calP|_{T_u}$-multicut such that for every $(s,t) \in Q$,
  $M$ intersects the $(s,t)$-path in $T$.

  Indeed, consider a terminal pair $(s,t) \in \calP|_{T_u} \cup Q$.
  If $(s,t) \in \calP|_{T_{u_i}}$ for some $i \in [\lambda]$,
  then $M$ intersects the $(s,t)$-path in $T$
  since $S^*_i$ is a $\calP|_{T_{u_i}}$-multicut by definition.
  Similarly, if $(s,t) \in O_i$ for some $i \in [\lambda]$,
  then $M$ intersects the $(s,t)$-path in $T$
  since $S^*_i$ intersects this path by definition.
  Thus, let us assume that
  $(s,t) \in \Pairs_{u,\pi}$.
  Suppose first that $\{s,t\} \cap V(T_{u_i}) \neq \emptyset$ for some $i \in [\lambda]$,
  say $s \in V(T_{u_i})$ without loss of generality.
  If $t \in V(T_{u_j})$ for some $j \in [\lambda]$, then $(p_i,p_j) \in \Pairs'_{u,\pi}$
  and so, $M$ intersects the $(s,t)$-path in $T$
  since $M^*$ intersects the $(p_i,p_j)$-path in $T'_u$ by definition.
  If $t \in V(T_u) \setminus \bigcup_{j \in [\lambda]} V(T_{u_j})$,
  then $(p_i,t) \in \Pairs'_{u,\pi}$
  and so, $M$ intersects the $(s,t)$-path in $T$
  since $M^*$ intersects the $(p_i,t)$-path in $T'_u$ by definition.
  Finally, if $t \in V(T) \setminus V(T_u)$,
  then $(p_i,u) \in \Pairs'_{u,\pi}$
  and so, $M$ intersects the $(s,t)$-path in $T$
  since $M^*$ intersects the $(p_i,u)$-path in $T'_u$ by definition.

  Second, suppose that $\{s,t\} \cap V(T_{u_i}) = \emptyset$ for every $i \in [\lambda]$.
  If $s,t \in V(T_u)$, then $M$ intersects the $(s,t)$-path in $T$
  since $M^*$ intersects this path by definition.
  Otherwise, exactly one of $s$ and $t$ belongs to $V(T_u)$,
  say $s \in V(T_u)$ without loss of generality,
  in which case $(s,u) \in \Pairs'_{u,\pi}$
  and so, $M$ intersects the $(s,t)$-path in $T$
  since $M^*$ intersects the $(s,u)$-path in $T'_u$ by definition.
  Thus, $M$ is as claimed.

  By minimality of $S$, it follows that $\wt(M) \geq \wt(S)$;
  but $\wt(M) = \wt(M^*) + \sum_{i \in [\lambda]} \Tab[u_i, O_i]$,
  as for every $i \in [\lambda]$, $\Tab[u_i, \cdot]$ is correctly filled,
  and thus, the converse inequality holds as well.
  Combined with the above, we conclude that the algorithm $\Aheavyv$ is correct.

  Finally, let us argue that the algorithm $\Aheavyv$ runs in
  $2^{d\ell} \cdot 2^{\OO(\ell^2\log \ell)} \cdot n^{\OO(1)}$-time.
  First, Step \ref{step:algoHeavy:Nuempty} can be done in \timewmctrees-time by Theorem~\ref{thm:wt-mc-trees-l} 
  since, in this case, $T_u$ has at most $\ell$ leaves.
  Otherwise, observe that for a fixed distribution,
  the most computationally demanding step is the call to $\Aunmc$
  in Step~\ref{step:algoHeavy:computingSolution},
  which takes \timewmctrees-time by Theorem~\ref{thm:wt-mc-trees-l},
  as $T'_u$ has at most $\ell$ leaves.
  Now note that there are at most $2^{d\ell}$ distributions to consider:
  indeed, $|N_u| \leq \ell$ and for every $i \in [\lambda]$,
  $|2^{O[u_i]}| \leq 2^d$ as $u_i$ is $d$-light.
  Thus, Step~\ref{step:algoHeavy:loop} takes $2^{d\ell} \cdot 2^{\OO(\ell^2\log \ell)} \cdot n^{\OO(1)}$-time
  in total and so, the algorithm indeed runs in the stated time.
\end{proof}

\section{Future Directions}\label{sec:conclusion}

The natural question to ask is, whether the running times of our algorithms can be improved.
Faster algorithms for the arcless instances 
in~\cite{kim2021directed},
directly yield faster algorithms for \wmctrees\ parameterized by the number of leaves.
Another interesting question is to determine the parameterized complexity of the (bi-objective) \wMCshort\ problem
with respect to structural parameters such as the number of leaves. 
There it seems difficult to use the flow augmentation technique from~\cite{kim2021directed},
since such a step takes exponential time in the solution size,
but the number of the leaves in the input may be much smaller than the solution size.

Another interesting follow up question is to determine if one can use the directed flow augmentation to resolve the parameterized complexity of
{\sc Weighted Steiner Multicut} on trees,
where given a tree $T$, sets $P_1, \ldots, P_r \subseteq V(T)$ each of size $p \geq 1$,
a weight function $\wt : V(T) \to \mathbb{N}$ and
positive integers $\ww,\kk$,
the goal is to determine if there exists a set $S \subseteq V(T)$ such that
$|S| \leq \kk$, $\wt(S) \leq \ww$ and for each $i \in [q]$
there exists $u_i,v_i \in P_i$ such that $T-S$ has no $(u_i,v_i)$-path.
Observe that \bimctrees\ is a special case of this problem when $p=2$.


\subsubsection*{Acknowledgements.}
Research supported by the European Research Council (ERC) consolidator grant No.~725978 SYSTEMATICGRAPH.
Philipp Schepper is part of Saarbrücken Graduate School of Computer Science, Germany.

\bibliographystyle{plainnat}
\bibliography{references}

\end{document}